\documentclass[sigplan]{acmart}

\usepackage[T1]{fontenc}
\usepackage[utf8]{inputenc}

\usepackage{listings, fancyvrb}
\usepackage{amsmath, amsthm}
\usepackage{enumitem}
\usepackage{graphicx}
\usepackage{subfig}
\usepackage{color,xcolor}
\usepackage{url}

\usepackage[font=small,labelfont=bf]{caption}
\usepackage[compact]{titlesec}
\usepackage{enumitem}
\renewcommand\footnotetextcopyrightpermission[1]{}

\usepackage[framemethod=TikZ]{mdframed}
\mdfdefinestyle{listingstyle}{%
  outerlinewidth=0.25pt,outerlinecolor=black,%
  innerleftmargin=5pt,innerrightmargin=5pt,innertopmargin=0pt,innerbottommargin=0pt%
}

\definecolor{amaranth}{rgb}{0.9, 0.17, 0.31}
\definecolor{ored}{rgb}{1.0, 0.27, 0.0}
\definecolor{ao}{rgb}{0.0, 0.5, 0.0}
\definecolor{ys}{rgb}{0.5, 0.0, 0.5}

\def\submission{1}

\if\submission 1
\newcommand{\inappendix}[1]{Section~#1 of the supplementary material}
\else
\newcommand{\inappendix}[1]{Appendix~#1}
\fi

\def\omitcomments{1}

\if\omitcomments 1
\newcommand{\Naama}[1]{}
\newcommand{\Guy}[1]{}
\newcommand{\Yihan}[1]{}
\newcommand{\Youla}[1]{}
\newcommand{\Y}[1]{}
\newcommand{\Hao}[1]{}
\newcommand{\Eric}[1]{}
\newcommand{\Note}[1]{}
\newcommand{\here}[1]{{}}
\newcommand{\hedit}[1]{#1}
\newcommand{\y}[1]{#1}
\newcommand{\er}[1]{#1}
\newcommand{\gb}[1]{#1}
\newcommand{\yh}[1]{{#1}}
\else
\newcommand{\y}[1]{{\color{blue} #1}\normalcolor}
\newcommand{\er}[1]{{\color{amaranth} #1}\normalcolor}
\newcommand{\gb}[1]{{\color{olive} #1}\normalcolor}
\newcommand{\yh}[1]{{\color{ys} #1}\normalcolor}
\newcommand{\Guy}[1]{\noindent
  {\color{blue}{\textbf{Guy: }#1}}}
\newcommand{\Hao}[1]{\noindent
  {\color{ao}{[[\textbf{Hao: }#1]]}}}
\newcommand{\Naama}[1]{\noindent
  {\color{cyan}{\textbf{Naama: }#1}}}
\newcommand{\Yihan}[1]{\noindent
  {\color{ys}{\textbf{Yihan: }#1}}}
\newcommand{\Youla}[1]{\noindent
  {\color{ored}{\textbf{Youla: }#1}}}
\newcommand{\Eric}[1]{\noindent
  {\color{magenta}{\textbf{Eric: }#1}}}
\newcommand{\Note}[1]{\noindent
	{\color{orange}{\textbf{Note: }#1}}}
\newcommand{\Y}[1]{\noindent
  {\color{ored}{#1}}}
  \newcommand{\hedit}[1]{{\color{ao}{#1}}}
\newcommand{\here}[1]{{\bf [[[ #1 ]]]}}
\fi

\newcommand{\Approach}{Approach}
\newcommand{\approach}{approach}

\newcommand{\anapproach}{an approach}
\newcommand{\ds}[1]{\texttt{#1}}
\newcommand{\bst}{\ds{NBBST}}

\newcommand{\bbst}{\ds{BST-64}}
\newcommand{\chromatic}{\ds{CT}}
\newcommand{\bchromatic}{\ds{CT-64}}
\newcommand{\vbst}{\ds{VcasBST}}
\newcommand{\vbbst}{\ds{VcasBST-64}}
\newcommand{\vchromatic}{\ds{VcasCT}}
\newcommand{\vbchromatic}{\ds{VcasCT-64}}
\newcommand{\pbst}{\ds{PNB-BST}}
\newcommand{\bpbst}{\ds{PNB-BST}}
\newcommand{\kiwi}{\ds{KiWi}}
\newcommand{\lfca}{\ds{LFCA}}
\newcommand{\kst}{\ds{KST}}
\newcommand{\snaptree}{\ds{SnapTree}}
\newcommand{\epochtree}{\ds{EpochBST}}

\newcommand{\nullv}{null}

\newcommand{\hl}[1]{\colorbox{cyan!30}{#1}}

\newcommand{\CAS}{CAS}
\newcommand{\vCAS}{versioned CAS}
\newcommand{\VCAS}{Versioned CAS}

\newcommand{\solo}{solo}
\newcommand{\Solo}{Solo}

\newcommand{\sololy}{solo}

\newcommand{\tsvalid}{valid}
\newcommand{\tsinvalid}{invalid}

\newcommand{\true}{\texttt{true}}
\newcommand{\false}{\texttt{false}}

\newcommand{\val}{\texttt{val}}
\newcommand{\ts}{\texttt{ts}}

\newcommand{\TBD}{\texttt{TBD}}

\newcommand{\vnext}{\texttt{nextv}}
\newcommand{\vnode}{\texttt{VNode}}
\newcommand{\initTS}{\texttt{initTS}}
\newcommand{\VHead}{\texttt{VHead}}

\newcommand{\op}[1]{\texttt{#1}} 

\newcommand{\findop}{\op{find}}
\newcommand{\Findop}{\op{Find}}
\newcommand{\insertop}{\op{insert}}
\newcommand{\rangeop}{\op{range}}
\newcommand{\Insertop}{\op{Insert}}
\newcommand{\deleteop}{\op{delete}}
\newcommand{\Deleteop}{\op{Delete}}

\newcommand{\findif}{\op{findif}}
\newcommand{\succr}{\op{succ}}
\newcommand{\multisearch}{\op{multisearch}}

\newcommand{\vrd}{\op{vRead}}
\newcommand{\vcas}{\op{vCAS}}
\newcommand{\takess}{\op{takeSnapshot}}
\newcommand{\readss}{{\op{readSnapshot}}}

\newcommand{\ovrsnap}{\op{OptreadSnapshot}}
\newcommand{\vrsnap}{\op{readSnapshot}}

\newcommand{\OptVCAS}{{\sc OptVerCAS}}
\newcommand{\VerCAS}{{\sc VerCAS}}

\newcommand{\myparagraph}[1]{\textbf{\emph{#1}}}

\newcommand{\rd}{\op{read}}

\newcommand{\cas}{\op{CAS}}

\newcommand{\Snapshot}{Camera}
\newcommand{\snapshot}{camera} 

\newcommand{\complex}{multi-point}
\newcommand{\Complex}{Multi-point}

\newcommand{\var}[1]{\texttt{#1}}

\newcommand{\Search}{\op{Search}}
\newcommand{\RSum}{\op{RangeSum}}
\newcommand{\RSTraverse}{\op{RSTraverse}}

\newcommand{\enqueue}{\op{enqueue}}
\newcommand{\Enqueue}{\op{Enqueue}}
\newcommand{\dequeue}{\op{dequeue}}
\newcommand{\Dequeue}{\op{Dequeue}}

\newcommand{\Clean}{\mbox{\sc Clean}}
\newcommand{\Mark}{\mbox{\sc Mark}}
\newcommand{\flag}{\mbox{\sc Flag}}
\newcommand{\NBBST}{\mbox{\sc NBBST}}
\newcommand{\MSQ}{\mbox{\sc MS-Queue}}
\newcommand{\VMSQ}{\mbox{\sc Ver-Queue}}
\newcommand{\VBST}{\mbox{\sc Ver-BST}}

\newcommand{\scan}{\op{scan}}
\newcommand{\peekHeadandTail}{\mbox{\tt peekEndPoints}}
\newcommand{\PeekHeadandTail}{\mbox{\tt PeekEndPoints}}

\newtheorem{theorem}{Theorem}
\newtheorem{lemma}[theorem]{Lemma}
\newtheorem{claim}[theorem]{Claim}
\newtheorem{observation}[theorem]{Observation}

\newtheorem{proposition}[theorem]{Proposition}
\theoremstyle{definition}
\newtheorem{definition}[theorem]{Definition}

\newtheorem{construction}[theorem]{Construction}

\newcounter{results}

\newcommand{\hide}[1]{}
\newcommand{\remove}[1]{}

\makeatletter
\lst@Key{countblanklines}{true}[t]%
{\lstKV@SetIf{#1}\lst@ifcountblanklines}

\def\StartLineAt#1{\lstset{firstnumber=#1}}

\lst@AddToHook{OnEmptyLine}{%
  \lst@ifnumberblanklines\else%
  \lst@ifcountblanklines\else%
  \advance\c@lstnumber-\@ne\relax%
  \fi%
  \fi}
\makeatother

\begin{document}


\title[Constant-Time Snapshots with Applications]
{Constant-Time Snapshots with Applications to
  Concurrent Data Structures}

\author{Yuanhao Wei}
\email{yuanhao1@cs.cmu.edu}
\affiliation{%
  \institution{Carnegie Mellon University}
  \city{Pittsburgh}
  \state{PA}
  \country{USA}
}

\author{Naama Ben-David}
\email{bendavidn@vmware.com}
\affiliation{%
  \institution{VMware Research}
  \city{Palo Alto}
  \state{PA}
  \country{USA}
}

\author{Guy E. Blelloch}
\email{guyb@cs.cmu.edu}
\affiliation{%
  \institution{Carnegie Mellon University}
  \city{Pittsburgh}
  \state{PA}
  \country{USA}
}

\author{Panagiota Fatourou}
\email{faturu@csd.uoc.gr}
\affiliation{%
  \institution{FORTH ICS and University of Crete}
  \country{Greece}}

\author{Eric Ruppert}
\email{ruppert@cse.yorku.ca}
\affiliation{%
  \institution{York University}
  \city{Toronto}
  \state{ON}
  \country{Canada}
}

\author{Yihan Sun}
\email{yihans@cs.ucr.edu}
\affiliation{%
  \institution{University of California, Riverside}
  \city{Riverside}
  \state{CA}
  \country{USA}
}

\renewcommand{\shortauthors}{Y. Wei, N. Ben-David, G.E. Blelloch, P. Fatourou, E. Ruppert, and Y. Sun}


\begin{abstract}
	We present 
	\gb{\anapproach{}} for efficiently taking snapshots of the state of a collection of CAS objects.
	\er{Taking a snapshot allows later operations to read the value that each CAS object had at the time
	the snapshot was taken.  Taking a snapshot requires a constant number of steps and returns a handle to the snapshot.
	Reading a snapshotted value of an individual CAS object using this handle is wait-free, taking time
	proportional to the number of successful CASes on the object since the snapshot was taken.}
	\er{Our fast, flexible snapshots yield simple, efficient implementations of} atomic \y{multi-point}  queries on concurrent data structures built from CAS objects.
	For example, in a search tree where child pointers are updated using CAS,
once a snapshot is taken,
one can atomically search for ranges of keys, find the first key that matches some criteria, or check if a collection of keys are all present, simply
by running a standard sequential algorithm on a snapshot of the tree.

To evaluate the  performance of our approach, we 
apply it to two 
search trees, one balanced and one not.
	\y{Experiments} show that the overhead \er{of supporting snapshots}
is low across a variety of workloads.
Moreover, in almost all cases, range queries on the trees built from our snapshots perform as well as or better than 
state-of-the-art concurrent data structures that support atomic range queries.
\end{abstract}

\maketitle
\pagestyle{plain}


\section{Introduction}
\label{sec:intro}

The widespread use of multiprocessor machines for large-scale
computations has underscored the importance of efficient concurrent
data structures.  Unsurprisingly, there has been significant work in
recent years on designing practical lock-free and wait-free data
structures to meet this demand and guarantee system-wide
progress.
Many
applications that use concurrent data structures require querying
large portions or multiple parts of the data structure. For example,
one may want to filter all elements by a certain property, perform
range queries, or simultaneously query multiple locations. However,
 such ``\complex{}'' queries have been notoriously hard to
implement efficiently.
Although it is easy to support \complex\ queries by locking large or multiple parts of the data structure,
this approach lacks parallelism.  Some concurrent data structures resort to \complex\ queries that provide no
guarantee of atomicity \cite{javaweak,NGPT15b}.
Other efforts have implemented specific queries
(e.g., range queries, iterators)~\cite{FPR19,AB18,BA12,Cha17,PT13,ALRS17,FNP17}.

\er{A general way to support \yh{efficient \complex\ queries} is to
provide the ability to take a \emph{snapshot} of the data structure.
Conceptually, a snapshot saves a read-only version of the state of the data structure
at a single point in
time~\cite{AADGMS93,Fic05,EllenFR07,FatourouFR06,FatourouK07,Jayanti05,BrodskyF04,Anderson94,perelman2010maintaining,Attiya11,Kumar14,BBSW19,FernandesC11}.
\Complex\ queries can be performed by
taking a snapshot and reading the necessary parts of that version to answer the query, while updates run concurrently.}
Snapshots are also used 
in database systems for
multiversioning and recovery~
~\cite{Reed78,BG83,papadimitriou1984concurrency,Postgres12,SQL13,neumann2015fast,Wu17},
and in persistent sequential data
structures~\cite{sarnak1986planar,driscoll1989making,DST91}.
\er{However, known approaches for
taking snapshots either limit the programming model (e.g. purely
functional~\cite{dickerson2020adapting, BBSW19}), use locks with no \yh{progress guarantees}~\cite{BG83,neumann2015fast,Kumar14}),
or are lock- or wait-free but
have large running times
~\cite{AADGMS93,FatourouK07,FernandesC11,Jayanti05,BrodskyF04}.}

We present an efficient algorithm to take snapshots of the state of a collection
of compare\&swap (CAS) objects\footnote{
\y{A CAS object $V$ stores a value and supports two atomic {\em operations}.
$V$.\rd() returns the value of $V$.
$V$.CAS($\mathit{old},\mathit{new}$)  compares the
value of $V$ to $\mathit{old}$ and if they are equal, it
changes the value of $V$ to ${\mathit{new}}$ and returns \true; otherwise,
it returns \false\ without changing $V$'s value.}}.
Our interface is based on creating a \emph{\snapshot} object  that
has a collection of associated \emph{\vCAS\ objects}, which
\er{support read and CAS instructions like normal CAS objects.
The \snapshot\ object supports a single operation \takess{} that
takes a snapshot of
the values stored in all the associated \vCAS\ objects, returning
a handle to the snapshot.}
Given a versioned CAS object $O$ and a snapshot handle \ts{} \er{obtained from the associated \snapshot\ object},
$O$.\readss{}(\ts) returns the value $O$ had at the time the handle was acquired by a \takess.
\er{New \vCAS\ objects   can be associated with an existing \snapshot\ object,
so our construction is applicable to dynamically-sized data structures.}

\er{Our interface is more flexible than the one traditionally used for a \emph{snapshot object}
\cite{AADGMS93}, which stores an array and provides \emph{update} operations that write to individual
components and \emph{scan} operations that return the state of the entire array.
Instead of creating a copy of the state of the entire shared memory in the local memory of a process,
our \takess\ simply makes it possible for a process
to later read  \emph{only} the memory locations it needs from shared memory, knowing that the collection of all such reads will be atomic.}
Although partial snapshot objects~\cite{AGR08,imbs2012help} allow scans of part of the array, they require the set of locations to be specified in advance, whereas our approach allows the locations to be chosen dynamically as the query is executed.

Our algorithm has the following important properties.

\begin{enumerate}[topsep=1.5pt, partopsep=0pt]
\item
\er{Taking} \y{a snapshot of the current state and returning a handle to it} takes constant time
(i.e., a constant number of instructions).

\item A CAS or read \er{of the current state of a \vCAS\ object} 
	takes constant time.  
		\y{Therefore,} \er{adding snapshots to a CAS-based data structure preserves the data structure's asymptotic time bounds.}

\item Reading the value of a \vCAS\ object from a snapshot takes time proportional
  to the number of successful CAS operations on the object since the snapshot.
		\y{Thus, a}ll reads are \y{{\em wait-free} (i.e., every  read is completed within
		a finite number of \er{instructions}.)}

	\item The algorithm is implemented using single-word read and \CAS, which are supported by modern architectures.
It does, however, require an unbounded counter.
\end{enumerate}
\hedit{We know of no previous general mechanism for snapshotting the state of memory that satisfies even the first two properties.}



Similarly to previous work~\cite{Reed78,BG83,sarnak1986planar,driscoll1989making,Kumar14,neumann2015fast,Wu17},
we\linebreak use a version
list for each CAS object.
The list has one node per update
(successful CAS) on the object.
Each node contains the value stored by the
update and a \emph{timestamp} indicating when the update occurred.  The list is
ordered by timestamps, most recent first.   The difficulty in
implementing version lists without locks \er{is} the need to add a node
to the version list, read a global timestamp, and save that
timestamp in the node, all atomically.    An important contribution of
our work is the mechanism used  to make these three steps appear atomic.

\myparagraph{Snapshots and \Complex{} Queries.}  Our interface provides a
simple way of converting \y{a concurrent data structure} built out of CAS
objects into one that supports snapshots: simply replace all
CAS objects with \vCAS\ objects that are all associated with a single \snapshot\ object.  If all shared mutable state is stored in the CAS
objects, then taking a snapshot will effectively \er{provide access to} \y{an} atomic copy of
\y{the entire} state of the data structure at the snapshot's linearization point\footnote{\y{We use the standard
definition of \emph{linearizability}~\cite{herlihy1990linearizability}, which roughly states that
every operation must appear to have taken effect atomically at its linearization point, between
its invocation and response.   
}}.
After taking a snapshot, a read-only query is free to visit
any part of the data structure state
at its leisure, even as updates proceed concurrently.
\gb{Often, the query can just be a standard sequential query executed on the snapshot.}

In Section~\ref{sec:query}, we describe how this can be used for arbitrary
queries on Michael-Scott queues~\cite{MS96}, Harris's
linked-lists~\cite{Harris01}, and two different binary search trees~\cite{EFRB10,BER14}. \Yihan{If we move any of them to the appendix, don't forget to change the sentence here.}
On the binary search trees, for example, one can
support atomic queries for finding the smallest key that matches a
condition, reporting all keys in a range, determining the \er{height} of the
tree, or \yh{multi-searching for a set of keys in the tree.}
\Eric{I deleted "average" before height.  Average height of a tree doesn't make sense.  Average depth of nodes in a tree, or I guess average height of nodes in a tree.}
\hedit{The time complexity of each query is the sequential cost of the query plus the number of \vcas{} operations it is concurrent with.}
In \er{our supplementary material}
we define more precisely when
\complex\ queries are possible using snapshots.

\myparagraph{Optimizations.}
Our algorithm introduces only constant overhead for existing operations, and allows the implementation of wait-free queries.
\er{However,} our construction does introduce a level of indirection: to access the value of a  \vCAS\ object, one must first access a pointer to the head of the version list, which leads to the actual value. This may introduce an extra cache miss per access. We therefore consider
\er{optimizations to avoid this in Section \ref{sec:opt}.}
The first optimization removes the versioning for \CAS{} objects that are never accessed by queries.
%
The second optimization applies to concurrent data structures that satisfy our \emph{recorded-once} property and avoids a level of indirection.
Roughly speaking, recorded-once means that each data structure node is the new value of a successful \cas\ at most once.
This allows us to store information for maintaining the version lists (in particular the timestamp and the pointer to the next older version) directly in the nodes themselves, thus removing a level of indirection.
This optimization can be applied to many lock-free data structures.

\myparagraph{Implementation and Experiments.}
\er{To study}
\y{
the overhead of our approach, we  applied it}
to two existing concurrent \y{binary search} trees, one
balanced and one not~\cite{EFRB10,BER14}.
\er{Adding support for snapshots}
 was very easy and required adding fewer than 150 lines of code in C++.
The experiments demonstrate that
the overhead is small. For example, it is about 9\% for
a mix of updates and queries on the current version of the tree.
We also
compare to state-of-the-art data structures that support atomic range queries,
including \kiwi~\cite{BBBGHKS17}, \lfca~\cite{WSJ18}, \pbst{}~\cite{FPR19}, and \snaptree{}~\cite{BCCO10}.
In almost all cases, \er{our data structure} performs as well as or better than all of these
special-purpose structures even though our approach is general purpose.    Finally,
we implement a variety of other atomic \complex{} queries
and show that the overhead compared to non-atomic
implementations, which are correct only when there are no concurrent
updates, is small.  Our implementation uses
epoch-based garbage collection~\cite{fraser2004practical}.

\y{
	\myparagraph{Contributions.} In summary, the paper's contributions are:} 
\begin{itemize}[topsep=1.5pt, partopsep=0pt]
	\item \er{A} simple, constant-time 
          \approach{} to take a snapshot of \er{a collection of CAS objects}.
	\item \er{A technique} to use snapshots to implement \er{linearizable} \complex{}
          queries on many lock-free  data structures.
	\item \er{Optimizations} that make the technique more practical.
	\item \y{Experiments showing our technique has low overhead, often outperforming
		other state-of-the-art approaches,} \hedit{despite being more general.}
\end{itemize}

\Yihan{It seems that we still don't have a summary of time bounds we achieve for ordered sets with range queries. Do we still want them?}
\Hao{I added the following sentence to the intro: "The time complexity each query would be the sequential cost of the query plus the number of \vcas{} operations it is concurrent with." Do we need to go in more detail?}
\here{Youla: Check that the notation is consistent with later sections, cas, read, etc.}


\section{Related Work}

There has been a long history of having
transactions see a snapshot of the state while other
transactions make updates.   This is often referred to as
multiversioning~\cite{Reed78,BG83,papadimitriou1984concurrency,perelman2010maintaining,Postgres12,SQL13,Kumar14,neumann2015fast,Wu17,sun2019supporting,BBSW19,riegel2006lazy,cachopo2006versioned}.
Indeed, the idea of version lists for snapshots dates back to Reed's
thesis on transactions~\cite{Reed78}.   This work is all applied to
transactions and none of it provides the theoretical guarantees described
in this paper.

\er{Implementing a snapshot object} is a classic problem in shared-memory computing
with a long history.  Fich surveyed some of this work~\cite{Fic05}.
A \emph{partial snapshot} object allows operations that take a snapshot of
selected entries of the array instead of the whole array \cite{AGR08, imbs2012help}.
An \emph{$f$-array} \cite{Jay02} is another generalization
of snapshot objects that allows
a query operation that returns the value of a function $f$ applied to a snapshot of
the array.
\er{As mentioned above, snapshot objects have a less flexible interface than our approach to snapshotting.}


\er{We describe in Section \ref{sec:query} how to use our snapshots to support \complex\ queries on a wide variety of data structures.
Previous work has focused on supporting such queries}
on \emph{specific} data structures.
Bronson {\it et al.}\cite{BCCO10} gave a blocking implementation of AVL trees
that supports a \scan\ operation that returns a snapshot of the whole data structure.
Prokopec {\it et al.}~\cite{PBBO12} gave a \scan\ operation for a hash trie by making the trie persistent:  updates copy the entire branch of nodes that they traverse.
\op{Scan} operations have also been implemented for non-blocking queues \cite{NGPT15a,NGPT15b,prokopec2015snapqueue}
and deques \cite{FNP17}.
Kallimanis and Kanellou \cite{KK15} gave a dynamic graph data structure that allows atomic dynamic traversals of a path.

\emph{Range queries}, which return all keys within a given range, have been studied for various implementations of ordered sets.
Brown and Avni \cite{BA12} gave an obstruction-free range query algorithm for $k$-ary search trees.  Avni, Shavit and Suissa \cite{ASS13} described how to support range
queries on skip lists.
Basin {\it et al.}~\cite{BBBGHKS17} described a concurrent implementation of
a key-value map that supports range queries.  Like our approach, it uses multi-versioning
controlled by a global counter.

Fatourou, Papavasileiou and Ruppert~\cite{FPR19} described a
persistent implementation of a binary search tree that permits
wait-free range queries, also based on version lists.
Our work borrows some of these ideas, but avoids the cumbersome handshaking and
helping mechanism  they use to synchronize between
scan and update operations.
This more streamlined approach makes our approach
easier to generalize to other data structures.
Winblad, Sagonas and Jonsson \cite{WSJ18} also gave a concurrent binary search tree that supports range queries.

Some researchers have also taken steps towards the design of general techniques
for supporting \complex{} queries
that can be applied to classes of data structures, \er{although none are as general as our approach}.
\Yihan{I dropped ``rather than ... data structures'' since we are short of space again.}

Petrank and Timnat \cite{PT13} described how to add a non-blocking snapshot operation
to non-blocking data structures such as linked lists and skip lists that implement
a set abstract data type.  Updates and scan operations must coordinate carefully
using auxiliary \emph{snap collector} objects.
Agarwal {\it et al.}~\cite{ALRS17} discussed what properties a data structure must
have in order for this technique to be applied.
Chatterjee \cite{Cha17} adapted Petrank and Timnat's algorithm to produce partial snapshots.

Arbel-Raviv and Brown \cite{AB18} described how to implement range queries for
concurrent set data structures that use epoch-based memory reclamation.
They assume
that one can design a traversal algorithm that is guaranteed to visit every item in the given
range that is present in the data structure for the entire lifetime of the traversal.
It is also assumed that updates are linearized at a write or CAS instruction, and
that the location of this instruction is known in advance.



\section{\VCAS{} Objects}
\label{sec:vcas}


\remove{
\Naama{Slightly shorter alternative version of the first paragraph of this section:
We begin by defining two new objects, a \emph{\vCAS{}} object and a \emph{\snapshot{}} object.
The \vCAS{} object behaves similarly to a regular \CAS{} object, but ``saves'' previous versions.
Each \vCAS\ object is associated with one \snapshot\ object, and supports three operations; \vrd{}, \vcas{} and \readss{}.
\vrd{} and \vcas{} function as normal read and CAS primitives, reading and updating  the current value of the object.
A \snapshot{} object provides a single operation, \takess{}, for taking snapshots of the current state of all its associated \vCAS{} objects. \takess{} returns a \emph{handle} to a snapshot, which can then be given as input to \readss{} on an associated \vCAS{} object. The \readss{} returns the value of the \vCAS{} object at the time the snapshot handle was produced. }

We begin by defining
two new objects, a \emph{\vCAS{}} object and a \emph{\snapshot{}} object.
The \vCAS{} object behaves similarly to a regular \CAS{} object, but ``saves'' previous versions to support taking snapshots.
Each \vCAS\ object is associated with one \snapshot\ object \Naama{(but multiple \vCAS{} objects can share a single \snapshot{} object.)}.
A \snapshot{} object is a shared object that provides an interface for taking snapshots of the current state of all associated \vCAS{} objects.
The \vCAS{} object supports three operations, \vrd{}, \vcas{} and \readss{}.
As in a regular \CAS{} object, \vrd{} returns the current value of the object and \vcas($oldV$, $newV$) changes the object's value to $newV$ provided that the current value is equal to $oldV$.
The \snapshot{} object supports an operation \takess{} which returns a handle to a snapshot.
This handle serves as an identifier that can be used by later \readss\ operations to access
the value an associated \vCAS\ object had at the linearization point of the \takess.
}

%

We begin with a sequential specification of our objects.

\begin{definition} [\Snapshot{} and \VCAS\ Objects]
\label{def:vcas}
A \emph{\vCAS{}} object stores a \emph{value} and supports three operations, \vrd{}, \vcas{}, and \readss{}.
A \emph{\snapshot{}} object supports a single operation, \takess{}.
\gb{Each \vCAS{} object $O$ is associated with a single \snapshot{}
  object
   when it is created.}
Consider a sequential history of operations on a \snapshot{} object $S$ and the set $\Lambda_S$ of \vcas{} objects associated with it.
The behavior of
operations on $S$ and $O$
for all $O\in \Lambda_S$, is specified as follows.
  \begin{itemize}[leftmargin=*, topsep=1.5pt]\setlength{\itemsep}{0pt}
    \item An $O$.\vcas{}(\var{oldV}, \var{newV}) attempts to update the value of $O$ to \var{newV} and this update takes place if and only if the current value of $O$ is \var{oldV}.
    If the update is performed, the \vcas{} operation returns \true{} and is \emph{successful}. Otherwise, the \vcas{} returns \false{} and is {\em unsuccessful}.
    \item An $O$.\vrd{}() returns the current value of $O$.
    \item The behavior of 
	    \y{\readss{} and \takess{}} are specified simultaneously.
		  A precondition of calling $O$'s \readss{}($ts$) operation is that there must have been an earlier \y{$S$.\takess{}()} operation that returned the handle~$ts$.
		  For any \y{$S$.\takess{}()} operation $T$ that returns $ts$ and any $O$.\readss{}($ts$) operation $R$, $R$ must return the value $O$ had when $T$ occurred.
  \end{itemize}
\end{definition}
Multiple \takess{} operations on a \snapshot\ object $S$ may return the same handle, but Definition \ref{def:vcas} implies that two \takess\ operations can return the same handle $ts$ only if each associated \vCAS{} object has  the same value when these two \takess{} operations occurred.

\subsection{A Linearizable Implementation}
\label{sec:vcas-alg}

We give a linearizable implementation of \vCAS{} and \snapshot\ objects,
where \vcas{}, \vrd{} and
\takess{} can all be supported in \y{constant time.}
Our implementation is given in Algorithm~\ref{fig:vcas-alg}.

\BeforeBeginEnvironment{lstlisting}{\begin{mdframed}[style=listingstyle]}
\AfterEndEnvironment{lstlisting}{\end{mdframed}}

\renewcommand{\figurename}{Algorithm}
\begin{figure*}[!th]\small
\vspace{-.1in}
\begin{minipage}[t]{.42\textwidth}
  \begin{lstlisting}[linewidth=.99\columnwidth, numbers=left,frame=none]
class Camera {
  int timestamp;
  Camera() { timestamp = 0; }   @\label{line:ss-con}@
  int takeSnapshot() {
    int ts = timestamp;           @\label{line:read-time}@
    CAS(&timestamp, ts, ts+1);    @\label{line:inc-time}@
    return ts; } }@\codelineskip@
class VNode {
  Value val; VNode* nextv; int ts;
  VNode(Value v, VNode* n){
    val = v; ts = TBD; nextv = n; } };@\codelineskip@
class VersionedCAS {
  VNode* VHead;
  Camera* S;
  VersionedCAS(Value v, Camera* s){            @\label{line:vcas-con}@
    S = s;
    VHead = new VNode(v, NULL);      @\label{line:initialize-VHead}@
    initTS(VHead); }
  void initTS(VNode* n) {
    if(n->ts == TBD) {               @\label{line:init1}@
      int curTS = S->timestamp;       @\label{line:readTS}@
      CAS(&n->ts, TBD, curTS); } }   @\label{line:casTS}@
\end{lstlisting}
\end{minipage}\hspace{.3in}
\begin{minipage}[t]{.52\textwidth}
\StartLineAt{31}
\begin{lstlisting}[linewidth=.99\textwidth, numbers=left, frame=none]
  Value readSnapshot(int ts) {
    VNode* node = VHead;                     @\label{line:readss-head}@
    initTS(node);                            @\label{line:readss-initTS}@
    while(node->ts > ts) node = node->nextv; @\label{line:read-while}@
    return node->val; }                      @\label{line:readss-ret}@ @\codelineskip@
  Value vRead() {
    VNode* head = VHead;                  @\label{line:readHead}@
    initTS(head);                         @\label{line:read-initTS}@
    return head->val; } @\codelineskip@
  bool vCAS(Value oldV, Value newV) {
    VNode* head = VHead;                  @\label{line:vcas-head}@
    initTS(head);                         @\label{line:vcas-initTS}@
    if(head->val != oldV) return false;   @\label{line:vcas-false1}@
    if(newV == oldV) return true;         @\label{line:vcas-true1}@
    VNode* newN = new VNode(newV, head);  @\label{line:newnode}@
    if(CAS(&VHead, head, newN)) {         @\label{line:append}@
      initTS(newN);                       @\label{line:initTSnewN}@
      return true; }                      @\label{line:vcas-true2}@
    else {
      delete newN;                        @\label{line:scdelete}@
      initTS(VHead);                      @\label{line:vcas-initTSHead}@
      return false; } } };                   @\label{line:vcas-false2}@
  \end{lstlisting}
\end{minipage}
\vspace{-.1in}
	\caption{Linearizable implementation of a \snapshot{} object and a \vCAS{} object.
	}
\label{fig:vcas-alg}
\end{figure*}
\renewcommand{\figurename}{Figure}




\myparagraph{The \Snapshot{} Object.}
The \snapshot{} object behaves like a global clock for all \vCAS{} objects associated with~it.
It is implemented as a counter called \var{timestamp} that stores an integer value.
A \takess{} simply returns the current value $ts$ of variable \var{timestamp}
as the handle
and attempts to increment \var{timestamp} using a \cas.
If this \cas\ fails, it means that another concurrent \takess\ has
incremented the counter, so there
is no need to try again.
The handle will be used by future \readss{} operations to find the latest version of any \vCAS\ object that existed when the counter was incremented from $\mathit{ts}$ to $\mathit{ts}+1$.

\myparagraph{The \VCAS{} Object.}
Each \vCAS{} object is implemented as a \er{singly-}linked list (a \emph{version list}) that preserves all earlier values committed by \vcas{} operations, where each version is labeled by a timestamp read from the \snapshot's counter during the \vcas.
The list is ordered with more recent versions
closer to the head of the list.
A regular \vrd{} operation  just returns the version at the head of the list.
A successful \vcas{} adds a node to the head of the list.
\emph{After} the node has been added to the list, the value of the snapshot object's counter is  recorded
as the node's timestamp.
A \readss{}($ts$) traverses the version list and returns the value
in the first node with timestamp at most $ts$.

The \vCAS\ object stores a pointer \VHead\ to the last node \y{added to 
the} object's version list.
Each node in this list is of type \vnode\ and stores
\begin{itemize}[leftmargin=*,noitemsep,topsep=0pt]
\item a value \val{}, which is immutable once initialized,
\item a timestamp \ts{}, which is the timestamp of the successful \vcas\ that stored \val{} into the object, and
\item a pointer \vnext{} to the next \vnode\ of the list, which contains the next (older) version of the object.
\end{itemize}
The version list essentially stores the history of the object.

\myparagraph{Timestamps.}
We use a special timestamp \TBD{} (to-be-decided) as the default timestamp for any newly-created \vnode.
We note that \TBD{} is not a valid timestamp and must be substituted by a concrete value later, once the \vnode\ has been added to the version list.
\y{When a \vnode\ $x$  is  added to the version list, we call the \initTS{} subroutine
(Line \ref{line:init1}--\ref{line:casTS})
to assign it a  valid timestamp} \er{read from the \snapshot\ object's \var{timestamp} field}.
Once $x$'s timestamp changes from \TBD{} to something valid, it will never change again, because the \cas{} on Line~\ref{line:casTS}  succeeds only if the current value is \TBD{}.
This \initTS{} function can be performed either by the process that added $x$ to the list, or by another process that is trying to help.

\myparagraph{Implementing \readss(\var{ts}) and \vrd$()$.} The \\ \readss{} function returns the latest version of the \vCAS{} object with timestamp at most \var{ts}.
It first reads \VHead{} and helps set the timestamp of the \vnode\ that \VHead{} points to by calling \initTS{}.
The \readss{} then traverses the version list by following \texttt{nextv} pointers until it finds a version with timestamp smaller than or equal to \var{ts}, and returns the value in this \vnode.
The \vrd{} function 
looks only at \VHead, helps set the timestamp of the \vnode\ that \VHead{} points to, and returns the value in that \vnode.

\myparagraph{Implementing \vcas(\var{oldV, newV}).}
This operation first reads \VHead{} into a local variable \var{head}.
Then it calls \initTS{} on \var{head} to ensure its timestamp is valid.
If the value in the \vnode\ that \var{head} points to is not \var{oldV}, the \vcas{} operation fails and returns \false{} (Line \ref{line:vcas-false1}).
Otherwise, if \var{oldV} equals \var{newV}, the \vcas{} returns \true{} because nothing needs to be updated.
If \var{oldV} and \var{newV} are different, and the \vnode\ that \var{head} points to contains the value \var{oldV}, the algorithm attempts to add a new \vnode\ with value \var{newV} to the version list.
It first allocates a new \vnode\ \var{newN} (Line \ref{line:newnode}) to store \var{newV} and lets it point to \var{head} as its next version. It then attempts to
add \var{newN} to the beginning of the list by swinging the pointer \VHead{} from \var{head} to \var{newN} using a \cas{} (Line \ref{line:append}).
If \yh{successful,} 
it then calls \initTS{} on the new \vnode\ to ensure its timestamp is valid, and returns \true{} to indicate success. 
Before this call to \initTS\ terminates, a valid timestamp will have been recorded
in the new \vnode, either by this \initTS\ or by another operation helping the \vcas.

If the \cas{} on Line \ref{line:append} fails, then \VHead{} must have changed during the \vcas{} operation.
In this case, the new \vnode\ is not appended to the version list. The algorithm  deallocates the new \vnode\ (Line~\ref{line:scdelete}) and returns \false{}.
An unsuccessful \vcas{} is also responsible for helping the first \vnode\ in the version list acquire a valid timestamp.

{\myparagraph{Helping.} As mentioned, a \vrd{}, \readss{} and an unsuccessful \vcas{} all
\er{help (by calling \initTS) to ensure that the timestamp of the \vnode\ at the head of the version list
is valid before they return. This}
\y{is necessary to overcome the main difficulty in
implementing version lists without locks, i.e., making the following three steps appear atomic:
adding a node to the version list, reading a global timestamp, and recording a valid
timestamp in the node.} 
\er{(See also the discussion of correctness, below.)}

\myparagraph{Initialization.}
We assume that the constructor \y{(Line~\ref{line:ss-con})} for the \snapshot\ \y{object}
completes before the constructor \y{(Line~\ref{line:vcas-con})} for any
associated \vCAS\ object is invoked.
(In practice, one will often have just one global \snapshot\ object for all \vCAS\ objects used in a data structure.)

We require, as a precondition of any \readss($ts$) operation on a \vCAS\ object $O$, that
$O$ was created before the \takess\ operation that returned the handle $ts$ was invoked.
In other words, one should not try to read the version of $O$ in a snapshot that was
taken before $O$ existed. When we use \vCAS\ objects to implement a pointer-based data structure (like a tree or linked list),
this constraint will be satisfied naturally:  if we take a snapshot
of the data structure, and then try to traverse a sequence of pointers in it using
\readss\ instructions, we will never find a pointer to $O$ if $O$ did not exist when
the snapshot was taken.
\Guy{Commenting out since somewhat obvious: When we use \vCAS\ objects to implement a pointer-based data structure (like a tree or linked list),
this constraint will be satisfied naturally:  if we take a snapshot
of the data structure, and then try to traverse a sequence of pointers in it using
\readss\ instructions, we will never find a pointer to $O$ if $O$ did not exist when
the snapshot was taken.}



\myparagraph{Correctness.}
\er{Theorem~\ref{thm:vcas} states the algorithm's properties.}

\begin{theorem}
\label{thm:vcas}
\y{Algorithm~\ref{fig:vcas-alg}}
is a linearizable implementation of \vCAS{} and \snapshot\ objects.
	The number of instructions performed by \rd{}, \vcas{}, and \takess{} is constant, and
	the number of instructions performed by $O$.\readss($ts$) is proportional to the number of successful $O$.\vcas{} operations that have been assigned timestamps larger than $ts$ (this number is measured at the time the \y{\readss } reads \var{VHead}).
\end{theorem}


 \er{A complete proof of Theorem~\ref{thm:vcas} appears in \inappendix{\ref{sec:vcas-long-proof}}.
 Here, we just describe the linearization points used in that proof.}
 We say that a timestamp of a \vnode{} is {\em \tsvalid} at some point
 if the \ts{} field is not \TBD\ at that point, and {\em \tsinvalid} otherwise.

\remove{
\yh{
>>>>>>> Stashed changes
 We will use \var{VHead} and \var{timestamp} for $O$.\var{VHead} and $S$.\var{timestamp} with clear context.
 We say that a timestamp of a \vnode{} is {\em \tsvalid} at some configuration $C$
 if the \ts{} field is not \TBD\ at $C$, and {\em \tsinvalid} otherwise.
<<<<<<< Updated upstream
 The \emph{head} of the version list is the \vnode\ pointed to by \var{VHead}. Below we list the linearization point for each operation.

=======
We use \emph{version list} to refer to the list from the \vnode\ pointed to by \var{VHead} and following the \vnext{} pointers.
 The \emph{head} of the version list is the \vnode\ pointed to by \var{VHead}.

 The only way to modify a version list is the \cas\ at Line \ref{line:append},
 which adds a new \vnode\ to the version list.
 Before this happens, \initTS\ is called to ensure a valid timestamp in
 the old head of the version list.
 The correctness of \readss\ depends on ensuring
 that the timestamp associated with a value is current
 (i.e., in \var{S.timestamp}) at the linearization point of the \vcas\ that stored the value.
 Hence, we linearize a successful \vcas\ at the time
 when the successfully-stalled timestamp was read from \var{S}.
 Note that a \vnode\ $n$ can appear at the head of the version list
 before the \vcas\ that created $n$ is linearized.
 This is why any other operation that finds the head of a \vnode\ with an invalid timestamp
 calls \initTS\ to help install a valid timestamp in it before proceeding.
 Below we list the linearization point for each operation.
 }
>>>>>>> Stashed changes

}
 \begin{itemize}[leftmargin=*, topsep=1.5pt, partopsep=0pt]
   \item For a \vcas\ 
   operation $V$

   \begin{itemize}
   \item
   If $V$ performs a successful \cas\ on Line~\ref{line:append} adding a node $n$ to the version list, and $n$'s timestamp eventually becomes valid, then $V$ is linearized on Line~\ref{line:readTS} of the \var{initTS} method that makes $n$'s timestamp valid.
   \item
   Let $h$ be the value of \var{VHead} at Line \ref{line:vcas-head} of a \vcas{} operation.
   If $V$ returns on Line \ref{line:vcas-false1} or \ref{line:vcas-true1}, it is linearized either at Line \ref{line:vcas-head} if $h$'s timestamp is valid at that time, or the first step afterwards that makes $h$'s
    timestamp valid.
   \item
   If $V$ returns \false\ on Line \ref{line:vcas-false2},
   then $V$ failed its \cas{} on Line \ref{line:append}.  Thus, some other \vcas\ operation
   changed \var{VHead} after $V$ read it at Line \ref{line:vcas-head}.
   We linearize the \vcas\ immediately after the linearization point of the \vcas\ operation $V'$ that made the \emph{first} such change.
   If several \vcas\ operations that return on Line \ref{line:vcas-false2} are linearized immediately after $V'$, they can be ordered arbitrarily.
   \end{itemize}

   \item For a \vrd{} operation that terminates, let $h$ be the \vnode\ read from \var{VHead} at Line \ref{line:readHead}.
   The \vrd\  is linearized at Line \ref{line:readHead} if $h$'s timestamp is valid at that time, or at the first step afterwards that makes $h$'s timestamp valid.

   \item A \readss{} operation that terminates is linearized at its last step.

   \item For a \takess{} operation $T$ that terminates, let $\mathit{ts}$ be the value read from \var{timestamp} on line \ref{line:read-time}, $T$ is linearized when \var{timestamp} changes from $\mathit{ts}$ to $\mathit{ts}+1$.
 \end{itemize}

\y{
	Intuitively, the correctness of an \yh{$O.$\readss{}} operation
depends on ensuring
that the timestamp associated with a value is current
(i.e., in the \var{timestamp} field of the camera object $S$ associated with $O$) at the linearization point of the \vcas\ that stored the value in $O$.
\yh{ Hence, we linearize a successful \vcas\ at the time
 when the successfully installed timestamp was read from \var{S}.
 Note that a \vnode\ $n$ can appear at the head of the version list
 before the \vcas\ that created $n$ is linearized.
}
This is why any other operation that finds a \vnode\ with an invalid timestamp
at the head of the version list calls \initTS\ to help install a valid
timestamp in it before proceeding.
This helping mechanism is crucial to prove
that the linearization points described above are well-defined and within the intervals
of their respective operations.
}
\Eric{I like this parag.  Hao said he might also add something about the case where old=new.}
\Yihan{Shortened it a little bit. }

\section{Supporting Linearizable Wait-free Queries}
\label{sec:query}

\begin{table*}
  \centering \small
  \vspace{-.1in}
  \begin{tabular}{llll}
  \bf Original data structure & \bf Operation & \bf Our Time Bounds & \bf Parameters\\
  \hline
  Michael Scott Queue~\cite{MS96} & \op{i-th}$(i)$: & $O(i + c)$ & $c$: number of dequeues concurrent with\\
  & \op{enqueue}/\op{dequeue}: & same as original & \hspace{0.3cm}the query \\
  \hline
  Harris Linked List~\cite{Harris01} & \op{range}$(s,e)$: & $O(m + p + c)$ & $m$: number of keys in the linked list\\
  & \multisearch{}$(L)$:& $O(m + p + c)$ & $c$: number of inserts and deletes concurrent with  \\
  & \op{ith}$(i)$: & $O(i + p + c)$ & \hspace{0.3cm}the query \\
  & \op{insert}/\op{delete}/\op{lookup}: & same as original & \\
  \hline
  NBBST~\cite{EFRB10} and CT~\cite{BER14} & \op{successor}$(k)$ & $O(h + c)$ & $m$: number of keys in the BST \\
  & \multisearch{}$(L)$: & $O(|L|\times h + c)$ & $h$: height of tree. In the case of CT, $h \in O(\log{}(m) + p)$ \\
  & \op{range}$(s,e)$: & $O(h + K(s, e) + c)$ & $K(s,e)$: number of keys in BST between $[s,e]$\\
  & \op{height}$()$:& $O(m + c)$ & $c$: number of inserts, deletes, rotations concurrent with \\
  & \op{insert}/\op{delete}/\op{lookup}: & same as original & \hspace{0.3cm}the query\\
  \hline
\end{tabular}
  \caption{Time bounds for various operations on concurrent queues, lists, and BSTs using our snapshot \approach{}. The number of processes is denoted by $p$. Parameters such as the number of keys in the data structure are measured at the linearization point of the query operation. \Hao{I defined $c$ to be the number of concurrent inserts and deletes instead of the number of concurrent \vcas{} operations because using the latter, understanding how good our bounds are would require knowing how many \vcas{} operations are performed by each insert/delete.} \Yihan{changed to ``Our time bounds'' to save a line.}}\label{tab:time}\vspace{-.2in}
\end{table*}

We use \vCAS\ objects to extend a large class of concurrent data structures
that are implemented using reads and \cas\ primitives to support linearizable wait-free queries.
Our \approach{} is general enough to
allow transforming many \complex{} read-only operations on a sequential data structure
into linearizable queries on the corresponding concurrent data structure.
To achieve this,
we define the concept of a \emph{\solo{}} query, i.e., a query that only reads the
shared state, and once invoked, is correct if running to completion without any other process taking steps during its execution.
Intuitively, a \solo{} query is one that runs on a ``snapshot'' of the
data structure, and is typically just a standard sequential query.

The \approach\ works as follows.
Each \cas\ or \rd\ on a \CAS\ object is replaced by a \vcas\ or \vrd\
(respectively) on the corresponding \vCAS\ object, all of which are associated with one
\snapshot{} object.
To perform a \solo{} query operation $q$, a process $p$ first executes \takess\ on the \snapshot{} object, to obtain a handle~$ts$.
Then, for any \CAS\ object that $q$ would have accessed in the data structure, $p$ performs \readss($ts$) on the corresponding \vCAS\ object.
Intuitively, \takess{} takes a snapshot of shared state,
and \solo{} queries then run on this snapshot while other threads may be updating concurrently.
\Yihan{Changed ``Op'' and ``q'' to ``q'' and ``p'' to be consistent with later description.}

Not all queries for existing concurrent data structures are solo queries.
Herlihy and Wing~\cite{herlihy1990linearizability} describe a queue implementation in which the linearization order of the enqueue operations depends on future dequeue operations. For that algorithm, no \solo\ query is possible.
However, for most data structures it is straightforward to implement solo queries.
\yh{Here we give examples of several concurrent data structures that
support solo queries. A thorough treatment of the conditions under which solo queries are sufficient,
all the formalism for our \approach, necessary proofs, and more examples, are provided in the supplementary material.
}

\myparagraph{FIFO Queue.}
We first consider Michael and Scott's concurrent queue (MSQ)~\cite{MS96},
which supports atomic enqueue and dequeue, as well as finding the oldest and newest elements.
Our scheme additionally provides an easy atomic implementation of more powerful operations
 such as returning the $i$-th element, or all elements, etc.
The mutable locations in a MSQ consist of a \op{head} pointer,
a \op{tail} pointer, and the \op{next} pointer in each of a linked list of elements,
pointing from oldest to newest.  The \op{head} points indirectly to the oldest remaining element,
and the \op{tail} points to the newest element, or temporarily to one element behind the newest.
The newest element always contains a null next pointer.  After applying our \approach,
all these pointers become \vcas\ objects, and
a \takess{} operation, $op$, will atomically capture
the state of all of them.  Any query can then easily reconstruct the part of the queue state it requires.
For example, the $i$-th  query can start at the \op{head} and follow the list
(calling \readss{} on each node, using the handle returned by $op$)
until it reaches the $i$-th element in the queue.
We note that each \op{next} pointer in the linked list is only successfully updated once,
so each \readss{} of a \op{next} pointer takes constant time.
Therefore, for example, finding the $i$-th element (from the head) in
a queue takes time $O(i + c)$ where $c$ denotes
the number of successful dequeues between the read of the timestamp by $op$ and the read of the \op{head}.

\myparagraph{Sorted Linked List.}  Harris's data
structure~\cite{Harris01} maintains an ordered set as a sorted linked
list (HLL), and supports insertions, deletions, and searches.  Our
\approach\ adds atomic versions of
\complex{} query operations, such as range queries, finding the first
element that satisfies a predicate, or multi-searches (i.e., finding
if all or any of a set of keys is in the list).  To properly implement
concurrent insertions and deletions, HLL marks a node before
splicing it out of the list.  The mark is kept as one bit on
the pointer to the next list node.  Deletes are linearized at the
point the mark is set.  The mutable state is comprised of the
\op{next} pointers of each link, which contains the mark bit.   If
these are versioned, a \takess{} will capture the full state.  A query
can then just follow the snapshotted linked list from the head, using
\readss{} on every node; all marked nodes should be skipped.

\hedit{Time bounds for range query, multi-search and finding the $i$-th element are given in Table \ref{tab:time}.
Each insert or delete performs up to two successful \vcas{} operations and each successful \vcas{} potentially causes the query to traverse an extra version node.
So in the worst case, queries incur an additive cost of $c$.
Each query also incurs an additive cost $p$ because it could encounter up to $p$ marked nodes.
}

\myparagraph{Binary Search Trees.}  We now consider concurrent binary
search trees (BST). 
Many such data structures have been
designed~\cite{EFRB10,BER14,BCCO10,BA12,AB18,BBBGHKS17,WSJ18,pam,blelloch2016just}.
All the BST structures we looked into work with solo queries
allowing for powerful atomic queries of the same type as in HLL (e.g.,
range queries and multi-searches), but potentially much faster since
they can often visit a small part of the tree. Queries on the structure of the tree (e.g., finding its height) can also be made.
Here we consider two such trees (which are also used in our experiments in Section~\ref{sec:exp}): the unbalanced
non-blocking binary search trees (\bst{}) of Ellen
et al.~\cite{EFRB10}, and the balanced non-blocking chromatic tree (\chromatic) of Brown
et al.~\cite{BER14}. 

The \bst{} data structure is a unbalanced BST with the data stored at the leaves
and the internal nodes storing keys for guiding  searches.
Every insertion involves inserting an internal node and a leaf,
and similarly a delete will remove an internal node and a leaf.
The data structure uses lock-free locks, ``locking'' one or two nodes
on each insertion or deletion.
The locks are implemented by pointing to a descriptor of the ongoing operation, so other threads
can help complete the operation if they encounter a lock.  This makes the data structure lock-free.
The linearization point is at the pointer swing that splices an internal node (along with a child) in or out.
Therefore at any point in time the child pointers of the internal nodes fully define the contents of the tree.
If these child pointers are kept as \vCAS{} objects, then a snapshot will capture the required state.
The queries can ignore the locks, and therefore the lock pointers, although mutable, do not need to be versioned
(discussed further in Section~\ref{sec:opt}).

The chromatic tree (\chromatic{}) is a balanced BST that also stores its data at the leaves.
It is based on a relaxed version of red-black trees, with colors at each node
facilitating rebalancing.
Concurrent updates are managed similarly to the \bst{}.  In particular, updates are linearized at
a single CAS that adds or removes a key.  So, obtaining a snapshot of the tree's child pointers is sufficient to
run \complex\ queries.

\hedit{Any query $q$ on \bst{} or \chromatic{} will take time proportional to the number of nodes it visits
plus the write contention of~$q$ (i.e., the number of \vcas\ operations concurrent with $q$
on memory locations accessed by $q$).
This  assumes that $q$ performs \readss{} on each \vCAS{} object at most once.
This can easily be  ensured by maintaining a local view of the tree and  calling \readss{} only for \vCAS{} objects that are not yet in the local view.
For the bounds in Table \ref{tab:time}, it suffices to show that
the number of \vcas\ operations concurrent with $q$ is at most the number of inserts, deletes and rotations concurrent with $q$.
This is because each \vcas{} is either due to a rotation (only applies to \chromatic{}) or is the linearization point of an  insert or delete.
}

\hedit{Importantly, our snapshot approach maintains the time bounds of all the operations supported by the original data structure. (In the case of \bst{} and \chromatic{}, the original operations would be insert, delete, and lookup).} 


\section{Optimizations}
\label{sec:opt}

We now present several ways to optimize our snapshotting \approach\ (and therefore \complex{} queries on such snapshots). While practical, these optimizations are not fully general; for each optimization, we describe when it can be applied. We present these optimizations in terms of a concurrent data structure $D$ to which we add snapshots and use them to run queries from the set $Q$. For ease of notation, we denote by $D'$ the version of $D$ that also supports the queries in $Q$.

\myparagraph{\textbf{Reducing  the Number of \VCAS{} Objects.}}
The first optimization \yh{applies to} cases where the creation of version lists can be avoided.
The optimization is accomplished by leaving some of the \CAS\ objects of $D$ {\em unversioned},
i.e., by not replacing them  with \vcas\ objects.
We can apply the optimization to immutable fields and \CAS\ objects that are never
accessed by any query $Op\in Q$. 
For example, \hedit{in \NBBST\, the only mutable fields accessed by query operations are the left and right pointers, so all other fields can be left unversioned.}

\myparagraph{\textbf{Avoiding Indirection.}}
The second optimization applies to
\vcas\ objects \yh{that} store pointers to other nodes.
We assume that in $D$ (and therefore also in $D'$),
every operation accesses nodes of the data structure
through one (or more) immutable entry points (e.g., the pointer to the root in a BST).

We use a {\em history} to denote a sequence of instructions that can be executed by an algorithm
		starting from its initial state.	
A node is {\em recorded-once} in a history,
if a pointer to it is the \var{newV} parameter of a \emph{successful} \vcas\ (
on any \vCAS\ object)
 at most once.
$D'$ is {\em recorded-once}
if for every history of $D'$,
(1) every node used in the history is recorded-once \hedit{and (2) all \vcas\ operations with the same \var{newV} parameter also have the same \var{oldV} parameter}.

The optimization requires that $D'$ is a recorded-once implementation
and works as follows.
For each \vCAS\ object $O$ that stores a pointer to a node in $D'$,
instead of creating a new \vnode\ to store the version pointer
and the timestamp, the optimization stores this information directly in the node pointed to by $O$, 
thus avoiding the level of indirection introduced by \vnode s.
This requires expanding each node object with two extra fields.
In \inappendix{\ref{sec:opt-appendix}},
we show pseudocode for
the new version of a node
and provide pseudocode for \vrsnap, \vrd, and \vcas\
after applying the optimization.
We call the resulting implementation $D_{opt}$.


It appears that this optimization can be applied to concurrent data
structures for which at any point in time, every object has at most
one pointer to it.  Examples include tree data structure where
pointers go from parent to child, or singly-linked lists.  However,
this can involve slight modifications to the original concurrent
algorithm.  In particular, if a node is being pointed to by one object
and is being moved to be pointed to by another object then it would be
recorded more than once.  To avoid this,
the object can be copied and a pointer to
the new copy written into the new location.
This modification should be done with care to preserve correctness.
We apply this
transformation in our \bst{} implementation (Section~\ref{sec:implementation}).
An outline of correctness and pseudocode for this optimization is
given in Section~\ref{sec:opt-appendix} of the supplementary material.

\remove{
By inspection of the pseudocode for \VerCAS\ and \OptVCAS, we observe that there is a straightforward analogy
between the code executed on each line of \VerCAS\ and the code executed on the corresponding line of \OptVCAS\ (see Lines 19--55 of the algorithms).
%
This similarity 
makes it easy to prove that \OptVCAS\ is linearizable,
by following the same proof technique as for \VerCAS.
}


\section{Implementation}
\label{sec:implementation}

We implemented our snapshotting \approach{} in both Java and C++. Using it, we implement snapshottable versions of three existing data structures (see details below).  Our code uses the optimizations discussed in Section~\ref{sec:opt}. 
To apply our \approach{} on top of an implementation, our code \hedit{initializes} a \snapshot{} object, makes an object \emph{versionable} by adding a
\hedit{timestamp and a version pointer field}
to it, and
replaces the original shared mutable pointers to versionable objects by \vcas{} objects.
For the optimization to avoid indirection, one should also ensure the versionable objects are only recorded once. For each implementation, we use one global \snapshot{} object.

For our implementation in Java, we implemented \hedit{the} four queries in Table \ref{tab:complexqueries} as examples of \complex{} queries.
All these queries simply look up the snapshot of the tree, visiting what is needed.
They are all atomic \hedit{(i.e., linearizable)}.

\myparagraph{Base Data Structures.} We applied our snapshotting \approach{} to the two BST structures
described in Section~\ref{sec:query}, \bst{} and \chromatic{}, as well as a lock-free unbalanced BST from \cite{AB18}.
For the first two data structures, we used Brown's Java implementations~\cite{trevorimpl}.
The third was implemented in C++ and provided by the authors of \cite{AB18}.


\myparagraph{Batching.}
Previous work has shown that the performance of concurrent BSTs is improved by batching keys
in nodes~\cite{FPR19, WSJ18, BBBGHKS17, BA12}.  We therefore applied the same batching technique
from \pbst{}~\cite{FPR19} and \lfca{}~\cite{WSJ18} to our Java implementations,
storing up to 64 key-value pairs in each leaf \yh{(more details can be found in \cite{FPR19})}.
Our experiments indicate that batching \hedit{often} improves performance both on the original versions and our snapshotted versions.
We did not apply batching in our C++ implementation since it was also not used by the C++ implementation we compare \hedit{to} \cite{AB18}.

\myparagraph{Recorded-Once.}
The recorded-once requirement is satisfied by \chromatic{} and the BST from \cite{AB18}, so the modification was simply as described above.
However, the \bst{} does not satisfy it; the delete operation swings a pointer (via \cas{}) to a node that already exists in the data structure.
To avoid this, our implementation makes a copy of the node
and swings the pointer to this new copy instead.  This requires some extra marking
and helping steps to preserve correctness and lock-freedom.
Then, the modifications were as described above.

\myparagraph{Garbage Collection.}
In our Java code we use an epoch-based memory reclamation \cite{fraser2004practical} to disconnect nodes from version lists when no longer needed.
The Java garbage collector will then collect the old versions. 
For C++, we directly use the epoch-based garbage collector in the code from ~\cite{AB18}. 

\myparagraph{Names.}
We refer to the non-snapshotted Java implementations as \bbst{}, and \bchromatic{}, and our modified snapshotted versions as \vbbst{}, and \vbchromatic{}.  For C++ we refer to our modified snapshotted version as \vbst{}.
We plan to make our code publicly available via GitHub.


\begin{table*}
  \centering \small
  \vspace{-.1in}
  \begin{tabular}{lll}
  \bf Query & \bf Definition & \bf Parameters in Figure \ref{fig:complex}\\
  \hline
  \rangeop{}$(s,e)$: & All keys in range $[s,e]$ & \texttt{range256}: $e=s+256$\\
  \succr{}$(k,c)$ & The first $c$ key-values with key greater than $k$ & \succr\texttt{1}: $c=1$, or \succr\texttt{128}: $c=128$\\
  \findif{}$(s,e,p)$~\cite{findif}:&The first key-value pair in range $[s,e)$&\findif\texttt{128}: $p(k)=(k~\mathit{mod}~128~\mathit{is}~0)$ \\ 
  \multisearch{}$(L)$:& For a list of keys in $L$, return their values (\nullv{} if not found)&\multisearch\texttt{4}: $|L|=4$\\
  \hline
\end{tabular}
  \caption{The \complex{} queries and their parameters we use in experiments. \vspace{-.15in}}\label{tab:complexqueries}
\end{table*}


\section{Experimental Evaluation}
\label{sec:exp}

\newcommand{\insTrim}[1]{{\includegraphics[height=2.5mm,trim={0 3.5mm 0 3.5mm},clip]{#1}}}

\begin{figure*}
\advance\leftskip-0.5cm
\small
\begin{tabular}{ccc}
  \multicolumn{3}{l}{\insTrim{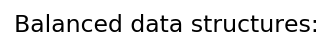}\insTrim{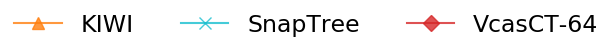}}\\
  \multicolumn{3}{l}{\insTrim{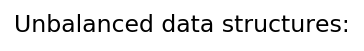}\insTrim{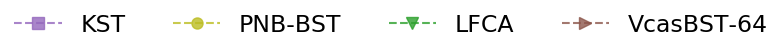}}\\
  \includegraphics[width=60mm]{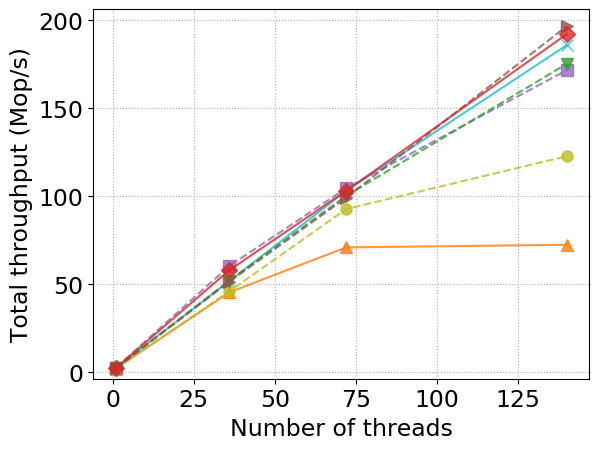}&
  \includegraphics[width=60mm]{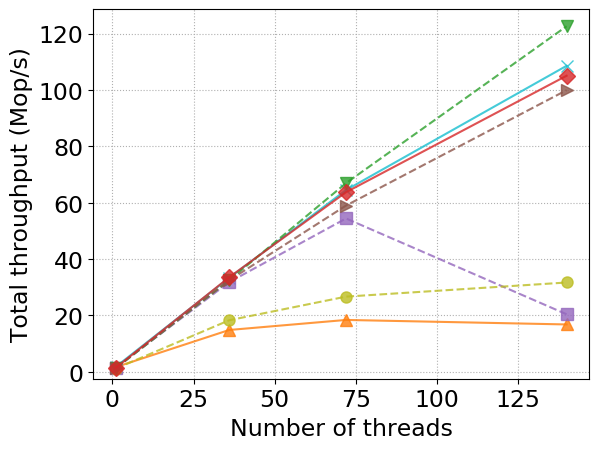}&
  \includegraphics[width=60mm]{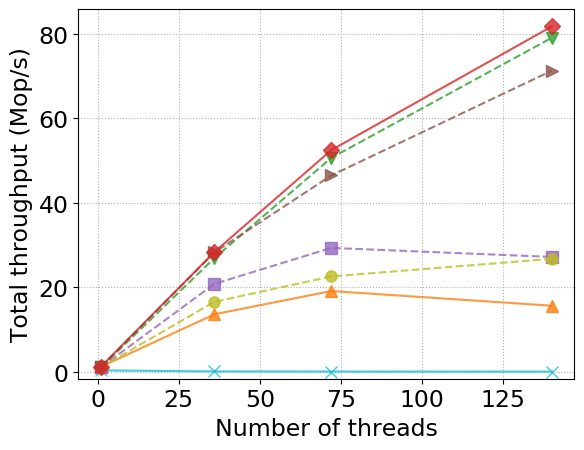}\\

    (a) Lookup heavy - 100K Keys:&
    (b) Update heavy - 100K Keys:&
    (c) Update heavy with RQ - 100K Keys:\\

    3\%ins-2\%del-95\%find-0\%rq&
    30\%ins-20\%del-50\%find-0\%rq&
    30\%ins-20\%del-49\%find-1\%rq-1024size\\

  \includegraphics[width=60mm]{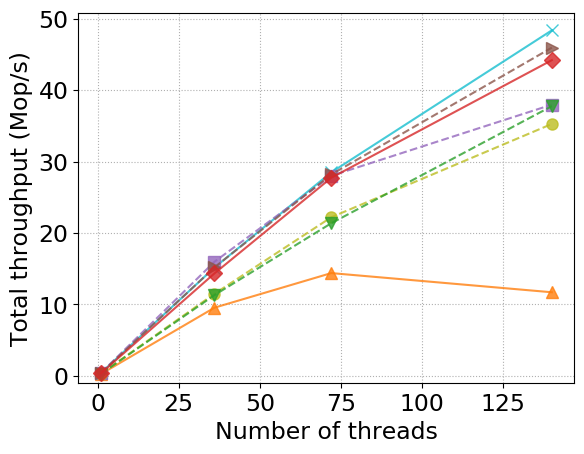}&
  \includegraphics[width=60mm]{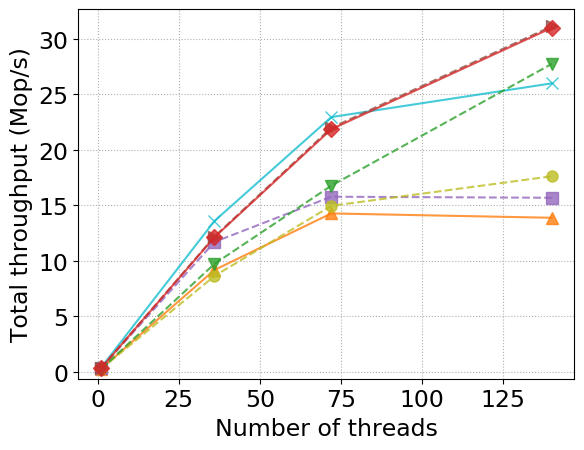}&
  \includegraphics[width=60mm]{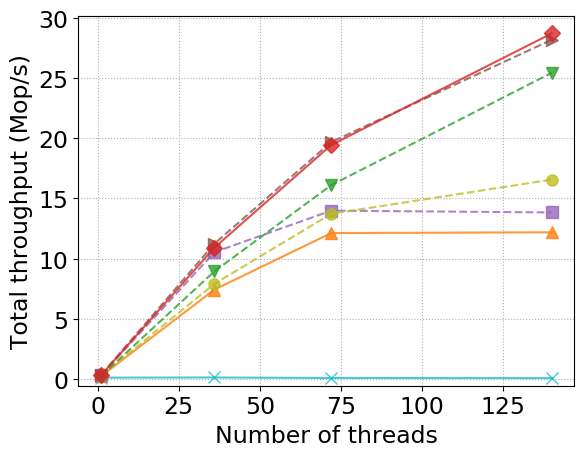}\\

    (d) Lookup heavy - 100M Keys:&
    (e) Update heavy - 100M Keys:&
    (f) Update heavy with RQ - 100M Keys:\\

    3\%ins-2\%del-95\%find-0\%rq&
    30\%ins-20\%del-50\%find-0\%rq&
    30\%ins-20\%del-49\%find-1\%rq-1024size\\

 \includegraphics[width=60mm]{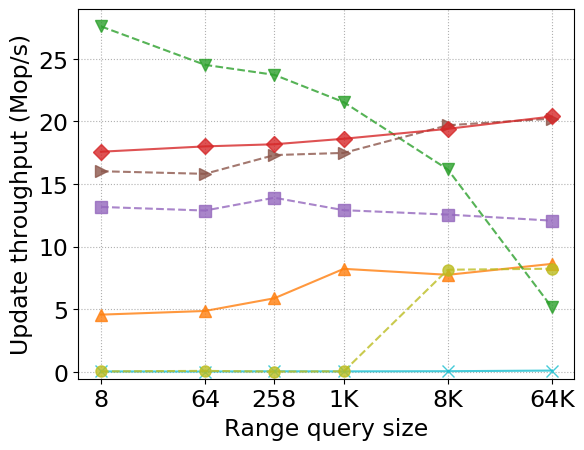}&
 \includegraphics[width=60mm]{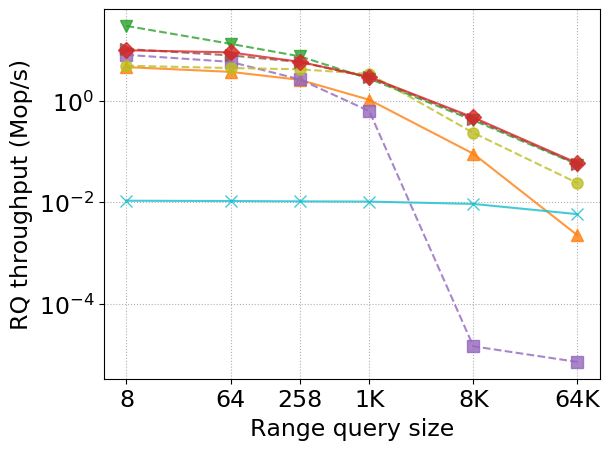}&
 \includegraphics[width=60mm]{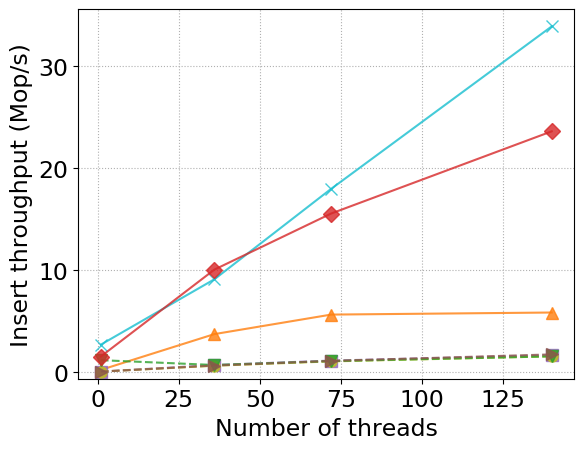}\\

    (g) Update Throughput - 100K Keys:&
    (h) RQ Throughput - 100K Keys:&
    (i) Insert Only, Sorted Sequence\\

    36 Update Threads, 36 RQ Threads&
    36 Update Threads, 36 RQ Threads&
    \\

 \includegraphics[width=60mm]{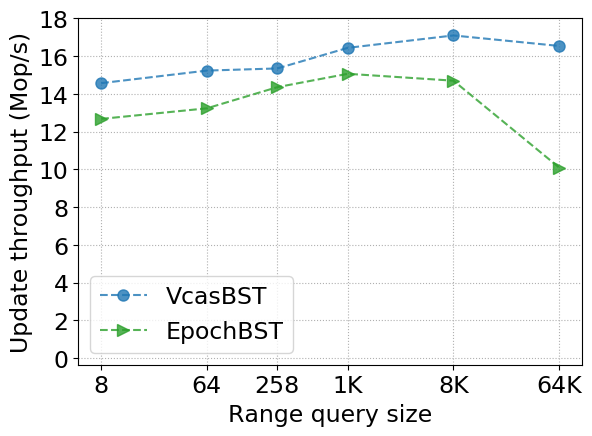}&
 \includegraphics[width=60mm]{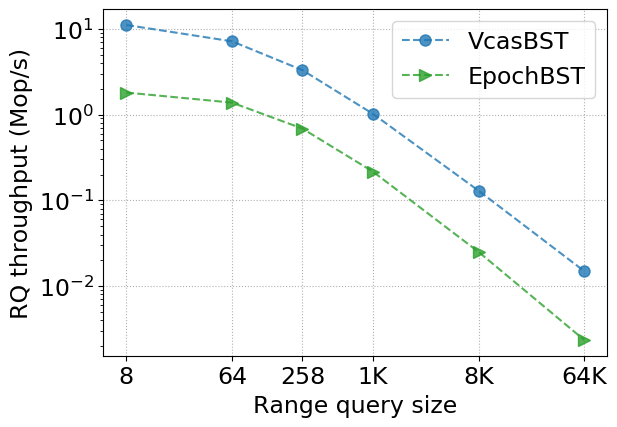}&
 \includegraphics[width=60mm]{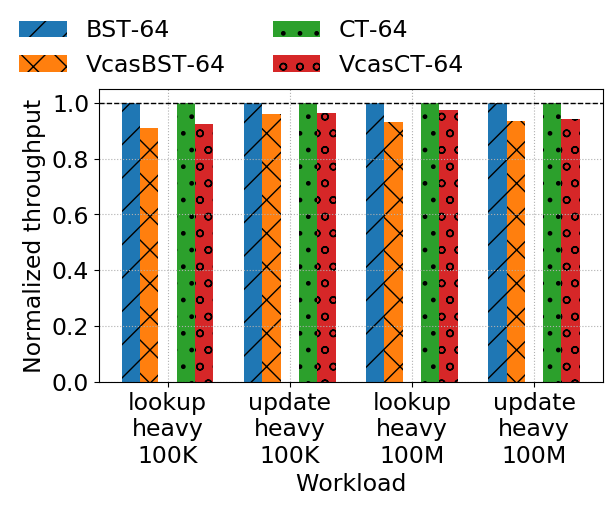}\\

    (j) [C++] Update Throughput - 100K Keys:&
    (k) [C++] RQ Throughput - 100K Keys:&
    (m) Overhead of Vcas, 140 threads,\\

    36 Update Threads, 36 RQ Threads&
    36 Update Threads, 36 RQ Threads&
    measured across various workloads\\
\end{tabular}
\caption{Results of experiments. }
\label{fig:exp}
\end{figure*}

\begin{figure}
\vspace{-.1in}
\includegraphics[width=65mm]{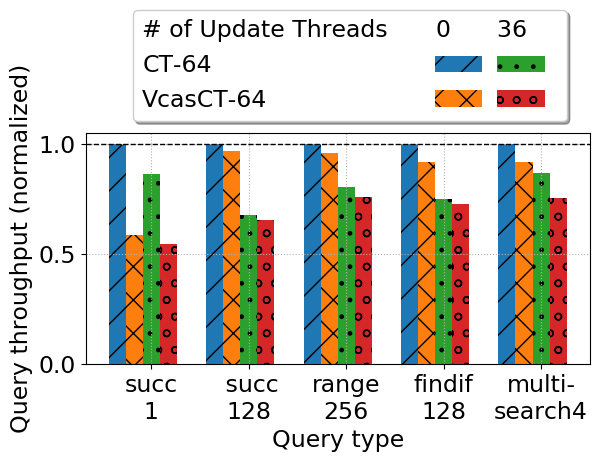}
\vspace{-.18in}
\caption{\small Throughput of atomic queries on \vbchromatic{} compared to non-atomic queries on \bchromatic{}, with or without update threads. Run using 36 query threads on a tree of size 100M.
\vspace{-.18in}
}
\label{fig:complex}
\end{figure} 


In this section, we provide experimental analysis, with two main goals: firstly, to understand the overhead our \approach{} introduces to concurrent data structures that originally did not support \complex{} queries, and secondly, to compare the performance of our \approach{} to state-of-the-art alternatives. That is, we compare our performance to that of concurrent dictionary data structures that support atomic range queries.
\remove{
\y{In this section, we provide our experimental analysis whose
goal is two-fold. Firstly, we aimed at understanding the overhead
incurred when applying our scheme on top of existing concurrent data structures.
Specifically, we study how much the approach slows down existing operations
on the current version, and the overhead of \complex{} queries in our scheme
compared to that of sequential queries.  \here{Y: The second part of the sentence
is not clear. I think we need to explain better what "sequential queries"  are
and how do we actually measure this overhead.}
Secondly, we compare the performance of our scheme to that of state-of-the-art
concurrent dictionary data structures that are specifically designed
to implement updates along with atomic range queries.
\here{Y: I removed the last sentence of this paragraph as we can always take a snapshot and
use it to answer any query (so I think the sentence was not true). }
}
\Yihan{Agree with Youla. Could we say that the other Java implementations we compared to do not support atomic the multipoint queries?}
}

\myparagraph{Other Structures that Support Range Queries.}
We compare with several state-of-the-art dictionary data structures: \snaptree{}~\cite{BCCO10},  \kiwi{}~\cite{BBBGHKS17},  \lfca{}~\cite{WSJ18},  \pbst{}~\cite{FPR19}, \kst{}~\cite{BA12}, and \epochtree{}~\cite{AB18} using code
provided by their respective authors.
Arbel-Raviv and Brown~\cite{AB18} presented several ways to add range queries to concurrent data structures, implemented in C++.
We use \epochtree{} to refer to their most efficient range queryable lock-free BST.
\hedit{Note that} \epochtree{} and \vbst{} add range queries to the same initial BST.
\remove{
We compare with several state-of-the-art dictionary data structures,
\kiwi{}~\cite{BBBGHKS17}, \lfca{}~\cite{WSJ18}, \pbst{}~\cite{FPR19}, \kst{}~\cite{BA12}, \snaptree{}~\cite{BCCO10}, and \epochtree{}~\cite{AB18}, as mentioned in Section \ref{sec:implementation}.
Code for these data structures were provided by their respective authors. \epochtree{} is in C++.
}
All the other data structures are written in Java.
They are all lock-free except \snaptree{} which uses fine-grained locking.
We classify \kiwi{}, \snaptree{}, and \vbchromatic{} as \emph{balanced} data structures
because they have logarithmic search time in the absence of contention,
and the others \y{as} \emph{unbalanced}.
For the $k$-ary tree (\kst{}), we use $k = 64$ which was shown to perform well across a variety of workloads \cite{BA12}.
We used batch size 64 for \vbbst{} and \vbchromatic{}, and we did the same for \lfca{} and \pbst{}.
This batch size has been shown to yield good range query performance for \lfca\ and \pbst\ in \cite{FPR19, WSJ18}.



\myparagraph{Setup.}
Our experiments ran on a 72-core Dell R930 with
4x Intel(R) Xeon(R) E7-8867 v4 (18 cores, 2.4GHz and 45MB L3 cache),
and 1Tbyte memory.  Each core is 2-way hyperthreaded giving 144
hyperthreads. We used
\texttt{numactl -i all} in all experiments, evenly spreading the
memory pages across the sockets in a round-robin fashion.
The machine runs Ubuntu 16.04.6 LTS.
The C++ code was compiled with g++ 9.2.1 with \texttt{-O3}.
Jemalloc was used for scalable memory allocation.
For Java, we \yh{used} OpenJDK 11.0.5 with flags \texttt{-server}, \texttt{-Xms300G} and \texttt{-Xmx300G}.
The latter two \y{flags} reduce interference from Java's GC.
\y{We} report the average of 5 runs, each of 5 seconds. For Java we also pre-ran 5 runs to warm up the JVM.
\yh{The variance is small in almost all tests.}


\myparagraph{Workload.}
We vary four key parameters: data structure size $n$, operation mix, range query size $\mathit{rqsize}$, and number of threads. 
In most experiments, we prefill a data structure with either $n=100K$ or $n=100M$ keys.
These sizes show the performance both when fitting and not fitting into the L3 cache.
Keys for each operation, and in the initial tree, are drawn uniformly at random from a range $[1, r]$,
where the integer $r$ is
chosen to maintain the initial size of the data structure.
For example, for $n=100K$ and a workload with 30\% inserts and 20\% deletes, we use \hedit{$r = n\times(30+20)/30 \approx 166K$.
}
\y{We} perform a mix of operations, 
represented by four values, $\mathit{ins}$, $\mathit{del}$, $\mathit{find}$, $\mathit{rq}$, which \er{are} the probabilities for each thread to execute an \insertop{}, \deleteop{}, \findop{}, or \rangeop{}, respectively.
\yh{Unbalanced trees can be balanced in expectation using uniformly random workloads,}
so we also run a workload with keys inserted in sorted order.
\Eric{Commented out definition of rqsize, since it is already defined at top of paragraph}
\Yihan{commented out some words and sentences to save space.}

\myparagraph{Scalability.}
Figures \ref{fig:exp}a-\ref{fig:exp}f show scalability (in Java) under a variety of workloads using thread counts 1, 36, 72, and 140.
Note that in Figures \ref{fig:exp}c and \ref{fig:exp}f, although range queries are only performed with 1\% probability, they still occupy a significant fraction of the execution time.

Generally, our two implementations (\vbchromatic{} and \vbbst{}) and \lfca{} have the best (almost-linear) scalability across all workloads.
\lfca{} outperforms our implementation in Figure \ref{fig:exp}b, but it is consistently slower in the 100M-key experiments (Figures \ref{fig:exp}d-f).
Snaptree is competitive with our trees in the absence of range queries, but it has no scalability with range queries due to its lazy copy-on-write mechanism.
\y{Overall, \vbchromatic{} is always among the top three algorithms and in most cases has the best performance. }



\myparagraph{Varying Range Query Size.}
We show the effect of \y{varying} range query size on
performance
in Figures \ref{fig:exp}g and \ref{fig:exp}h (Java), and  Figures \ref{fig:exp}j and \ref{fig:exp}k (C++).
36 dedicated threads ran range queries and 36 ran updates.
Each update thread performs 50\% inserts and 50\% deletes on a data structure initialized to 100K keys.
To better understand the cost of updates and range queries, we plot the throughput of each operation separately.

In Figure \ref{fig:exp}g, \bpbst{} has low update throughput when $\mathit{rqsize}\le 1024$.
This is because their update operations are forced to abort and restart whenever a new range query begins, and thus decreasing range query size lowers update throughput.
\kst{} performs decently in most workloads except when each range query covers a significant fraction of the key range,
where the update throughput is below 100 operations per second with $\mathit{rqsize}\ge 8K$.
This is because their range query performs a double collect of the desired range and is forced to restart if it sees an update in that range.

Data structures that increment a global timestamp with every range query become bottlenecked by this increment when range queries are frequent.
This applies to our trees as well as \bpbst{}, \kiwi{}, and \epochtree{}.
Consequently, with $\mathit{rqsize}=8$, \lfca{} has 3x faster range queries when compared to our trees (Figure \ref{fig:exp}h).
However, \lfca{} avoids using a global timestamp by having update operations help ongoing range query operations.
This helping becomes more frequent and more costly when $\mathit{rqsize}$ is large, as shown in Figure \ref{fig:exp}g.
For $\mathit{rqsize}=64K$ (about a third of the key range), the update throughput of our trees is 4x faster than \lfca{}. Other than \lfca{}, all the other implementations have mostly stable update throughput with varied range size, among which \vbchromatic{} has the best overall performance.

\y{Figures} \ref{fig:exp}j and \ref{fig:exp}k \y{compare} \y{the performance of the C++ version of} \vbst{} with \y{that of} \epochtree{}. Range queries on \vbst{} are 4.7--6.3x faster than \epochtree{}.
This is because a range query on \epochtree{} has to visit three nodes in the retired list for each concurrent \deleteop{}.
Thus, \epochtree{} visits 1.5--5.5x more nodes in range queries than \vbst{}. For updates, \vbst{} is at least \hedit{7\% faster than \epochtree{}, and up to 64\% faster on the largest range query size.}
\here{Y: It might be better to put the last four figures 2j - 2o and 3 together in one figure, i.e. display all four in a single row
(or if we have space, in two figures, one containing 2j and 2k and another containing 2o and 3).}
\Hao{We can try this if we need space.}



\myparagraph{Sorted Workload.}
In Figure \ref{fig:exp}i, we test  the Java implementations under a sorted workload.
We insert an array of sorted keys into an initially empty tree by splitting the array into chunks of size 1024 and placing the chunks on a global work queue; when a thread runs out of work, it grabs a new chunk from the head of the work queue.
As expected, the  balanced trees, \vbchromatic{}, \kiwi{} and \snaptree{}, outperform the unbalanced ones.
On 140 threads, \snaptree{} is 1.4x faster than \vbchromatic{}, which is in turn 4.1x faster than \kiwi{}.

\myparagraph{Overhead of Our \Approach{}.}
In Figure \ref{fig:exp}m, we compare our Java implementations \vbbst{} and \vbchromatic{} with the original data structures, \bbst{} and \bchromatic{}, using 140 threads.
\Hao{2l looks too much like 21, so I instead used 2m.}
The numbers in Figure \ref{fig:exp}m are normalized to the throughput of \bbst{} and \bchromatic{} to make the overheads easier to read.
The overall overhead of our \approach\ is low, ranging between 2.7\% and 9.1\% depending on the workload.
This overhead includes the time for epoch-based memory management and the cost of using \vcas{} and \vrd{}.
For \vbbst{}, it also includes the extra \y{actions we take to ensure that deletes} satisfy the recorded-once property.

We also \y{measure} the overhead of our \complex{} queries
\rangeop{}, \succr{}, \findif{}, and \multisearch{}, with parameters shown in Figure \ref{fig:complex}.
We compare the throughputs for \vbchromatic{} with \emph{non-atomic} \complex{} queries on the original \bchromatic{},
which  \y{simply run} their sequential algorithms (and are not linearizable).
Non-atomic \multisearch{}, for example, simply calls \findop{} for each  key.
Figure \ref{fig:complex} \y{shows the cost that our \approach\ has to pay to provide} query atomicity.

\here{Y: Table 1 should be moved here or be referenced earlier.}
\Yihan{It is also referred to in Section 6 so it's presented earlier. Moved it to somewhere in the middle.}
\Hao{I think the current location of table 1 is fine.}
\y{All queries other than \texttt{succ1}
exhibit low overhead:
they are between 2.9\% and 12.8\% slower than their non-atomic counterparts.
Having concurrent updates slows down both the atomic and non-atomic queries by about the same amount.
For \texttt{succ1}, our scheme exhibits larger overheads (36.8-41.4\%) due to the counter bottleneck
when the query size is too small.}


\myparagraph{\y{Summary.}}
\hedit{Overall, our snapshot \approach\ has low overhead and,}
\y{despite its generality},
performs well compared to existing special-purpose data structures.
In particular, \vbchromatic{} had the best overall performance \y{in most cases} among all the range queryable data structures we evaluated.
\Hao{I don't think we need ``in most cases'' because it does have the best overall performance.}
\vbbst{} is \y{also} competitive on uniform workloads.
\remove{None of the other implementations achieved to be consistently faster than (or even comparable to) \vchromatic{} across all workloads.
\remove{
The main bottleneck in our approach is the global timestamp, which appears in workloads with mixed update and small range queries. It would be interesting to look for efficient ways to circumvent this.}
\Yihan{Not sure if we should finish with the main bottleneck of our approach, but it is actually worth mentioning here. }
\here{Y: As I agree with Yihan, I have removed the discussion about the global timestamp from here. As far as I remember, we have already mentioned that our algorithms
have this problem earlier.}
\here{Y: Also, the last sentence still sounds kind of weak. It would be better if we finish with a stronger sentence.
If we cannot find something, I suggest to simply also remove the current last sentence.}
\Yihan{I agree. It would be nice if we can have a short summary that emphasis the strength of our implementation.}
}

\bibliographystyle{ACM-Reference-Format}
\bibliography{strings,biblio}


\clearpage

\appendix


\section{Detailed Proof of Correctness of \vCAS\ and \snapshot\ Objects}
\label{sec:vcas-long-proof}

\subsection{Linearization Points}

 Given a \snapshot\ object $S$ and a \vCAS\ object $O$ associated with it, in this section, we describe how their operations are linearized, but we defer the detailed proof of Theorem \ref{thm:vcas} to Appendix \ref{sec:vcas-proof}.

 To state the linearization points, we first introduce some useful terminology.
 When referring to the variables $O$.\var{VHead} and $S$.\var{timestamp}, we often abbreviate them to \var{VHead} and \var{timestamp}.
 We say that a \vnode{} has a {\em \tsvalid} timestamp at some configuration $C$
 if the value of its \ts{} field is not \TBD\ at $C$.
 Otherwise, the timestamp of the node is called {\em \tsinvalid}.
 We use the term \emph{version list} to refer to the list that results from starting at the \vnode\ pointed to by \var{VHead} and following the \vnext{} pointers.
 The \emph{head} of the version list is the \vnode\ pointed to by \var{VHead}.

 The only way to modify the version list is the \cas\ at Line \ref{line:append},
 which swings the \var{VHead} pointer to a new \vnode\ whose \var{nextv} pointer
 points to the previous head of the version list.
 This has the effect of adding the new \vnode\ to the beginning of the version list.
 Before this can happen, \initTS\ is called to install a valid timestamp in
 the old head of the version list.
 This ensures that the only \vnode\ in the version list with an invalid timestamp
 is the first one.  At the time a \vnode's timestamp becomes valid,
 it is therefore still at the head of the version list.


 The correctness of \readss\ operations depends on ensuring
 that the timestamp associated with a value is current
 (i.e., in \var{S.timestamp}) at the linearization point of the \vcas\ that stored the value in $O$.
 So, we linearize a \vcas\ that adds a \vnode\ to the version list at the time
 that the timestamp eventually written into that \vnode\ was read from \var{S.timestamp}.
 This means that there may be a \vnode\ at the head of the version list
 before the \vcas\ that created that \vnode\ is linearized.
 This is why any other operation that finds a \vnode\ with an invalid timestamp
 at the head of the version list calls \initTS\ to help install a valid
 timestamp in it before proceeding.
 This helping  mechanism is crucial in the argument 
 that all of the following linearization points are well-defined and within the intervals
 of their respective operations.

 \begin{itemize}
   \item A \vcas\ 
   operation is linearized depending on how it executes.

   \begin{itemize}
   \item
   If the \vcas\ performs a successful \cas\ on Line~\ref{line:append} that adds a node $n$ to the version list, and $n$'s timestamp eventually becomes valid, then the \vcas\ is linearized on Line~\ref{line:readTS} of the \var{initTS} method that makes $n$'s timestamp valid.
   Note that this case includes all \vcas\ operations that return \true\ at Line \ref{line:vcas-true2}.
   \item
   Let $h$ be the value of \var{VHead} at Line \ref{line:vcas-head} of a \vcas{} operation.
   If the \vcas{} operation returns on Line \ref{line:vcas-false1} or \ref{line:vcas-true1},  then it is linearized either at Line \ref{line:vcas-head} if $h$'s timestamp is valid at that time, or the first step afterwards that makes $h$'s
    timestamp valid.
   \item
   Finally, consider a \vcas(\var{oldV, newV}) operation $V$ that returns \false\ on Line \ref{line:vcas-false2}.
   This is the most subtle case.
   The return on Line \ref{line:vcas-false2} is only reached when $V$ fails its \cas{} on Line \ref{line:append} because some other \vcas\ operation
   changed \var{VHead} after $V$ read it at Line \ref{line:vcas-head}.
   We linearize the \vcas\ immediately after the linearization point of the \vcas\ operation $V'$ that made the \emph{first} such change.
   (If several \vcas\ operations that return on Line \ref{line:vcas-false2} are linearized immediately after $V'$, they can be ordered arbitrarily.)
   \end{itemize}

   \item For a \vrd{} operation that terminates, let $h$ be the \vnode\ read from \var{VHead} at Line \ref{line:readHead}.
   The \vrd\  is linearized at Line \ref{line:readHead} if $h$'s timestamp is valid at that time, or at the first step afterwards that makes $h$'s timestamp valid.

   \item A \readss{} operation that terminates is linearized at its last step.

   \item For \takess{} operations, let $t$ be the value read from \var{timestamp} on line \ref{line:read-time}.
   A \takess{} operation that terminates is linearized when the value of \var{timestamp} changes from $t$ to $t+1$.
 \end{itemize}

\subsection{Proof of Correctness}
\label{sec:vcas-proof}
In this section, we prove that Fig.~\ref{fig:vcas-alg} is a linearizable implementation of \vCAS{} and \snapshot{} objects.
First we argue that it suffices to prove linearizability for histories consisting of a single \vCAS{} object and a single \snapshot{} object.
Suppose two \vCAS{} objects are associated with different \snapshot{} objects.
Then we can prove linearizability for the two sets of objects independently because they do not access any common variables and do not affect each other in terms of sequential specifications.
Suppose two \vCAS{} objects $O_1$ and $O_2$ are associated with the same \snapshot{} object $S$.
Let $H'$ be a history of operations on these three objects.
Furthermore, let $H'_1$ be the history $H'$ restricted to only operations from $S$ and $O_1$, and similarly, let $H'_2$ be the history $H'$ restricted to only operations from $S$ and $O_2$.
We will define the linearization points of $S$ so that they are not affected by operations on $O_1$ or $O_2$.
Therefore, showing that both $H'_1$ and $H'_2$ are linearizable is sufficient for showing that $H'$ is linearizable because $S$ will be linearized the same way in both $H'_1$ and $H'_2$.

Let $H$ be a history of a \vCAS{} object $O$ and a \snapshot{} object $S$.
We assume that $S$ and $O$ are initialized by their constructors (Line \ref{line:ss-con} and \ref{line:vcas-con}, respectively) before the beginning of $H$.
We assume this history satisfies the precondition (described in Definition \ref{def:vcas}) that whenever \readss($ts$) is invoked, there must be a completed \takess{} operation that returned $ts$.
When referring to the variables $O$.\var{VHead} and $S$.\var{timestamp}, we will often abbreviate them to \var{VHead} and \var{timestamp}.

We first introduce some useful terminology.
We say that a \vnode{} has a {\em \tsvalid} timestamp at some configuration $C$
if the value of its \ts{} field is not \TBD\ at $C$.
Otherwise, the timestamp of the node is called {\em \tsinvalid}.
We use the term \emph{version list} to refer to the list that results from starting at the \vnode\ pointed to by \var{VHead} and following the \var{nextv} pointers.
The \emph{head} of the version list is the \vnode\ pointed to by \var{VHead}.

A \emph{modifying} \vcas\ operation is one that performs a successful \CAS\ on line \ref{line:append}.
Due to the if statement on line \ref{line:vcas-true1}, if \vcas($oldV$, $newV$) is a modifying \vcas{} operation, then $oldV \neq newV$.
Note that modifying \vcas{} operations can return only on line \ref{line:vcas-true2} and any operation that returns on line \ref{line:vcas-true2} is a modifying \vcas{}.
A \vcas\ is \emph{successful} if it is a modifying \vcas\ or if
it returns \true\ at line \ref{line:vcas-true1}.
Otherwise, it is \emph{unsuccessful}.

We first show that the only change to a version list is inserting a \vnode\ at the beginning of it.

\begin{lemma}
\label{lem:stay-in-list}
\Eric{New lemma to simplify proof of some others.}
Once a \vnode\ is in the version list, it remains in the version list forever.
\end{lemma}
\begin{proof}
The only way to change a version list is a successful \CAS\ at line \ref{line:append},
which changes \var{VHead} from \var{head} to \var{newN}.  When this happens,
$\var{newN->nextv} = \var{head}$, so all \vnode s that were in the version list before the
\CAS\ are still in the version list after the \CAS.
\end{proof}

\Eric{added next parag to show no segmentation faults; I thought we should explain where the precondition is needed.}
It is easy to check that every time we access some field of an object via a pointer to that object, the pointer
is not \var{NULL}.  \var{VHead} always points to a \vnode\ after it is initialized
on Line \ref{line:initialize-VHead} of $O$'s constructor.  It follows that every call to
\initTS\ is on a non-null pointer.
The precondition of \readss(\var{ts}) ensures that \var{ts} is a timestamp
obtained from $S$ after $O$ was initialized and is therefore greater than or equal
to the timestamp that $O$'s constructor stored in the initial \vnode\ of the version list.
Thus, the \readss\ will stop traversing the version list when it reaches that initial \vnode,
ensuring that \var{node} is never set to \var{NULL} on line \ref{line:read-while}.

\myparagraph{\textbf{Linearization Points.}}

Before we can define the linearization points, we need a few simple lemmas that
describe when \vnode s have valid timestamps.
We start with an easy lemma about \var{initTS}.

\begin{lemma}
\label{lem:initTS}
The following hold:
\begin{enumerate}
\item
\label{lem:head-before-initTS}
Before \initTS\ is called on a \vnode, \var{VHead} has contained a pointer to that \vnode.
\item
\label{lem:valid-after-initTS}
After a complete execution of \initTS\ on some \vnode, that \vnode's timestamp is valid.
\end{enumerate}
\end{lemma}

\begin{proof}
All calls to \initTS\ are done on a pointer that has either been read from \var{VHead} or
successfully CASed into \var{VHead}.
Once a timestamp is valid, it can never be modified again, since only a \CAS\ on line \ref{line:casTS} modifies the \var{ts} variable of any \vnode.
The \CAS\ on Line \ref{line:casTS} can fail only if the \var{ts} variable is already a
valid timestamp.
\end{proof}

\begin{lemma}
\label{lem:headInvalid}
In every configuration $C$, the only \vnode\ in the version list that can have an invalid
timestamp is the head of the version list.
\end{lemma}

\begin{proof}
\Eric{I rewrote this proof because the former argument was missing the fact that the pointer swing adds just a single node.  I tried to relate the proof more directly to the definition of version list.}
No \vnode's \var{nextv} pointer changes after the \vnode\ is created,
so the only way the version list can change is when \var{VHead} is updated.
Moreover, no \vnode's timestamp ever changes from valid to invalid.
So, we must only show that updates to \var{VHead} preserve the claim.

The value of \var{VHead} changes only when a successful \CAS\ is executed on Line~\ref{line:append}
  of an instance of \vcas. Consider any such successful \CAS\ by some process $p$
  and assume the claim holds in the configuration before the \CAS\ to show that it holds immmediately after the \CAS.
  This \CAS\ changes \var{VHead} from \var{head} to \var{newN}.
  By the initialization of \var{newN} on Line \ref{line:newnode}, that \vnode's \var{nextv} pointer is \var{head}.  So, we must show that  \var{head} and all \vnode s reachable from \var{head} by following \var{nextv} pointers have valid timestamps when the \CAS\ occurs.
  Before executing this \CAS, $p$ executes \var{initTS}(\var{head}), so, by Lemma~\ref{lem:initTS}(\ref{lem:valid-after-initTS}), that \vnode's timestamp is valid at the time that the \CAS\ is executed.
  Since the \CAS\ is successful, \var{VHead} was equal to \var{head} immediately before
  the \CAS, so all nodes reachable from that \vnode\ had valid timestamps, by our assumption.
\end{proof}

The next lemma is used to define the linearization point of a modifying \vcas.

\begin{lemma}
\label{lem:lin-head}
  Suppose an invocation of \initTS\ makes the timestamp of some \vnode\ $n$ valid.  Then, $n$ is the head of the version list when that \initTS\ executes Line \ref{line:readTS} and \ref{line:casTS}.
\end{lemma}

\begin{proof}
\Eric{I think this proof is simpler than what was written before.}
By Lemma \ref{lem:initTS}(\ref{lem:head-before-initTS}), every call to \initTS{} is on a pointer that has previously been in \var{VHead},
so $n$ has been in the version list before \initTS\ is called.
By Lemma \ref{lem:stay-in-list}, $n$ is still in the version list when Line \ref{line:readTS} and \ref{line:casTS} are executed.
By Lemma \ref{lem:headInvalid},
$n$ remains at the head of the version list until its timestamp becomes valid
when \initTS\ performs Line \ref{line:casTS}.
%
\end{proof}

We are now ready to define linearization points.  As we define them, we argue that the linearization point of each operation is well-defined and within the interval of the operation.  \Eric{I did it this way to satisfy Youla.  The arguments are pretty short, with the lemmas moved earlier, so I don't think it is too disruptive to the flow.}

\begin{itemize}
  \item A \vcas\ 
  operation is linearized depending on how it executes.

  \begin{itemize}
  \item
  If the \vcas\ performs a successful \CAS\ on Line~\ref{line:append} that adds a node $n$ to the version list, and $n$'s timestamp eventually becomes valid, then the \vcas\ is linearized on Line~\ref{line:readTS} of the \var{initTS} method that makes $n$'s timestamp valid.
  Lemma \ref{lem:lin-head} implies that the linearization point occurs after the \vcas\  adds $n$ to the version list at Line \ref{line:append}.
  If the \vcas\ terminates, it first calls \initTS\ on $n$ at line \ref{line:initTSnewN}, so
  Lemma \ref{lem:initTS}(\ref{lem:valid-after-initTS}) ensures the \vcas\ is linearized
  and that the linearization point comes before the end of that \initTS.
  \Eric{I changed preceding a little:  the previous version forgot to linearize \vcas\ operations that add a node to the version list, then die before returning true.  The wording is a bit tricky, because we want to linearize a \vcas\ that adds a node iff the node's timestamp becomes valid.}
  \Hao{I think this is a good change.}
  \item
  Let $h$ be the value of \var{VHead} at Line \ref{line:vcas-head} of a \vcas{} operation.
  If the \vcas{} operation returns on Line \ref{line:vcas-false1} or \ref{line:vcas-true1},  then it is linearized either at Line \ref{line:vcas-head} if $h$'s timestamp is valid at that time, or the first step afterwards that makes $h$'s
   timestamp valid.
  \Eric{I removed a clause from the defn of this linearization point that seemed redundant in view of Lemma \ref{lem:headInvalid}.  I changed it a bit again to make linearization points steps rather than configs because that makes it easier to define the value of a \vcas\ object in a config (if there could be lin points inside the config then the object could conceivably have multiple values in that config)}
  Lemma \ref{lem:initTS}(\ref{lem:valid-after-initTS}) ensures this step exists and is within the interval of the \vcas, since \var{initTS} is called on $h$ at line \ref{line:vcas-initTS}.
  \item
  Finally, consider a \vcas(\var{oldV, newV}) operation $V$ that returns \false\ on Line \ref{line:vcas-false2}.
  This is the most subtle case.
  The return on Line \ref{line:vcas-false2} is only reached when $V$ fails its \cas{} on Line \ref{line:append} because some other \vcas\ operation
  changed \var{VHead} after $V$ read it at Line \ref{line:vcas-head}.
  We linearize the \vcas\ immediately after the \vcas\ operation $V'$ that made the \emph{first} such change.
  (If several \vcas\ operations that return on Line \ref{line:vcas-false2} are linearized immediately after $V'$, they can be ordered arbitrarily.)

  To argue that this linearization point is well-defined, we must show that the \vnode\ $n$ that $V'$ added to the version list gets a valid timestamp, so that $V'$ is assigned a linearization point as described in the first paragraph above.
  By Lemma \ref{lem:stay-in-list}, $n$ is still in the version list when $V$ reads \var{VHead} at Line \ref{line:vcas-initTSHead}.
  If $n$ is no longer at the head of the version list, then $n$'s timestamp must be valid, by Lemma~\ref{lem:headInvalid}.
  Otherwise, if $n$ is still the head of the version list, then
  $n$'s timestamp is guaranteed to be valid after $V$ calls \initTS\ on $n$ (Line \ref{line:vcas-initTSHead}), by Lemma \ref{lem:initTS}(\ref{lem:valid-after-initTS}).
  So, in either case, $V'$ is assigned a linearization point, which is before the timestamp of $n$ becomes valid.  Thus, $V'$ (and therefore $V$) is linearized before the end of $V$.
  Lemma \ref{lem:lin-head} implies that the linearization point of $V'$ (and therefore of $V$) is after $V'$ adds $n$ to the version list, which is after $V$ reads \var{VHead}.
  This proves that $V$'s linearization point is inside the interval of $V$.
 \end{itemize}
  \item For a \vrd{} operation that terminates, let $h$ be the \vnode\ read from \var{VHead} at Line \ref{line:readHead}.
  The \vrd\  is linearized at Line \ref{line:readHead} if $h$'s timestamp is valid at that time, or at the first step afterwards that makes $h$'s timestamp valid.
  Lemma \ref{lem:initTS}(\ref{lem:valid-after-initTS}) ensures that this step exists and is during the interval of the \vrd,
  since the \vrd\ calls \var{initTS} on $h$ at Line \ref{line:read-initTS}.
  \Eric{Again, I made analogous changes to the vcas lin point.}

  \item A \readss{} operation that terminates is linearized at its last step.

  \item For \takess{} operations, let $t$ be the value read from \var{timestamp} on line \ref{line:read-time}.
  A \takess{} operation that terminates is linearized when the value of \var{timestamp} changes from $t$ to $t+1$.  We know that this occurs between the execution of Line \ref{line:read-time} and \ref{line:inc-time}:  either the \takess\ operation made this change itself if the \CAS\ at line \ref{line:inc-time} succeeds, or some other \takess\ operation did so, causing the \CAS\ on line \ref{line:inc-time} to fail.
\end{itemize}

Note that all operations that terminate are assigned linearization points.  In addition,
some \vcas\ operations that do not terminate are assigned linearization points.

\myparagraph{\textbf{Proof that Linearization Points are Consistent with Responses}}
Recall that $H$ is the history that we are trying to linearize.
In the rest of this section, we prove that each operation returns the same response in $H$ as it would if the operations were performed sequentially in the order of their linearization points.

\begin{lemma}
\label{lem:head-linearized}
Assume \var{VHead} points to a node $h$ in some configuration $C$.
If $h$.\var{ts} is valid in $C$ then either $h$ is the \vnode\ created by the constructor of $O$, or the \vcas\ that created $h$ is linearized before the configuration that immediately precedes $C$.
\end{lemma}
\begin{proof}
Suppose $h$.\var{ts} is valid in $C$ but $h$ is not the \vnode\ created by the
constructor of $O$.  Then $h$ is created by some \vcas\ operation $V$
that added $h$ to the head of the version list.
Since $h$.\var{ts} is valid in $C$,
some step prior to $C$ set $h$.\var{ts} by executing Line \ref{line:casTS}.
The linearization point of $V$ is at the preceding execution of Line \ref{line:readTS}.
Thus, the linearization point precedes the configuration before $C$.
\end{proof}

\Eric{I decided to go with a different definition of value of the \vCAS\ object
to make it easier to argue about all of the \vcas\ operations that fail.}
We define the \emph{value of the \vCAS\ object in configuration $C$} to be the
value that a \vCAS\ object would store if all of the \vcas\ operations linearized
before $C$ are done sequentially in linearization order (starting from the initial value
of the \vCAS\ object).
The following  crucial lemma  describes how the value of the \vCAS\ object
is represented in our implementation.
It also says that the responses returned by all \readss{} and \vcas\ operations are consistent
with the linearization points we have chosen.

\begin{lemma}
\label{lem:cur-val}
In every configuration $C$ of $H$ after the constructor of the \vCAS\ object has completed,
\begin{enumerate}
\item
if \var{VHead} points to the \vnode\ created by the constructor of the \vCAS\ object, then
 \var{VHead->val} is the value of the \vCAS\ object,
\item
if the linearization point of the \vcas\ that created the first node in the version
list is before $C$, then \var{VHead->val} is the value of the \vCAS\ object, and
\item
otherwise, \var{VHead->nextv->val} is the value of the \vCAS\ object.
\end{enumerate}
Moreover, each \vrd\ and \vcas\ operation that is linearized at or before $C$ returns the same result in $H$ as it would return when all operations are performed sequentially
in their linearization order.
\end{lemma}
\begin{proof}
We prove this by induction on the length of the prefix of $H$ that leads to $C$.
In the configuration immediately after the constructor of the \vCAS\ object terminates,  \var{VHead->val} stores the initial value of the \vCAS\ object.

Since \var{nextv} and \var{val} fields of a \vnode\ do not change after the
\vnode\ is created, we must only check that the invariant is preserved by steps
that modify \var{VHead} or are linearization points
of  \vcas\ operations (which may change the value of the \vCAS\ object) or \vrd\ operations.
We consider each such step $s$ in turn and show that, assuming the claim holds for the configuration $C$  before $s$, then it also holds for the configuration $C'$ after $s$.

First, suppose $s$ is a successful \CAS\ on \var{VHead} at line \ref{line:append} of a \vcas\ operation.
It changes \var{VHead} from \var{head} to \var{newN}, where $\var{newN->next} = \var{head}$.
By Lemma \ref{lem:headInvalid}, \var{head->ts} is valid when this \CAS\ occurs, since \var{head}
becomes the second node in the version list.
By our assumption, the value of the \vCAS\ object prior to the \CAS\ is $\var{head->val}$.
Since this step is not the linearization point of any \vcas\ operation, the value after the
\CAS\ is still \var{head->val}.
By Lemma \ref{lem:initTS}(\ref{lem:head-before-initTS}) \initTS\ is only called on a pointer that has been in \VHead\ previously,
and \var{newN} has never been in \var{VHead} before this \CAS, we know that
\var{newN->ts} is \TBD.  So the invariant holds after the \CAS, since $\var{VHead->nextv->val} = \var{head->val}$.

\Eric{These next 2 parags are the bits that I think were not clear in the old proof}

Now, consider a step $s$ that is the linearization point of a modifying \vcas(\var{oldV, newV}),
which we denote $V$, possibly
followed by the linearization points of some other \vcas\ operations that return \false\ on
Line \ref{line:vcas-false2}.
Since $V$ is a modifying \vcas, it added a new \vnode\ $n_1$ to the head of the version list
in front of node $n_2$.
This happens after $V$ checks that $n_2.\var{val} = \var{oldV} \neq \var{newV}$ on Line \ref{line:vcas-false1}--\ref{line:vcas-true1} and sets $n_1.\var{nextv}$ to point to $n_2$
and sets $n_1.\var{val}$ to \var{newV}  on Line \ref{line:newnode}.
By Lemma \ref{lem:lin-head}, $n_1$ is still the head of the version list when step $s$ occurs.
So in the configuration $C$ before $s$, the value in the \vCAS\ object is $n_2.\var{val} = \var{oldV}$, by our assumption that the claim holds in $C$.
Thus, when $V$ occurs in the sequential execution,
it returns \true\ and changes the value of the
\vCAS\ object to \var{newV}.
Note that $\var{VHead->val} = \var{newV}$ in $C'$.
It remains to check that all other \vcas\ operations that return \false\ at line
\ref{line:vcas-false2} and are linearized immediately after $V$
should return \false\ in the sequential execution and therefore
do not change the value of the \vCAS\ object.
Consider any such \vcas\ $V'$ of the form \vcas(\var{oldV',newV'}).
By the definition of the linearization point of $V'$,
$V$ makes the first change to \var{VHead} after $V'$ reads it on
Line \ref{line:vcas-head}.
So, $V'$ must have read a pointer to $n_2$ on Line \ref{line:vcas-head}.
Since $V'$ returns \false\ at Line \ref{line:vcas-false2}, it must have seen
$n_2.\var{val} =\var{oldV'}$ at Line \ref{line:vcas-false1}.
Thus, $\var{oldV'} = n_2.\var{val} = \var{oldV} \neq \var{newV}$,
so when each of the \vcas\ operations $V'$ is executed sequentially in
linearization order, it should return
\false\ and leave the state of the \vCAS\ object equal to \var{newV}.
The claim for $C'$ follows.

Finally, consider a step $s$ that is the linearization point of one or more \vrd\ operations or \vcas\ operations
that return at Line \ref{line:vcas-false1} or \ref{line:vcas-true1}.
Consider any such operation $op$.
Let $h$ be the node at the head of the version list when $op$ reads \var{VHead} at Line \ref{line:readHead} or \ref{line:vcas-head}.
Then $s$ is either this \rd\ or a subsequent execution of Line \ref{line:casTS} that
makes $h$'s timestamp valid.  Either way, \var{VHead} points to $h$ in $C'$, by Lemma
\ref{lem:lin-head}.
By Lemma \ref{lem:head-linearized}, either case (1) or (2) of the claim applies
to configuration $C$.
Either way, the value of the \vCAS\ object in $C$ is $h.val$.
If $op$ is a \vrd, then it returns $h.val$ as it should.
If $op$ is a \vcas\ that returns \false\ at Line \ref{line:vcas-false1}, it would do the same in the sequential execution in linearization order because $op$ reads
the state of the \vCAS\ object in $C'$ from $h.\var{val}$ on Line \ref{line:vcas-false1}
and sees that it does not match its \var{oldV} argument.
If $op$ returns \true\ at Line \ref{line:vcas-true1}, it would also return \true\ when performed in linearization order because the state of the \vCAS\ object in $C'$ matches both
$op$'s \var{oldV} and \var{newV} values.
In all cases the value of the \vCAS\ object does not change as a result of $op$, so it is still $h.\var{val}$ in $C'$,
and the invariant is preserved.
\end{proof}

The following observation follows directly from the way modifying \vcas{} operations are linearized.

\begin{observation}
\label{obs:ts-valid}
  Consider a \vnode\ $n$ that was added to the version list by a modifying \vcas{} $V$.
  If the timestamp of  $n$ is valid, then $n$.\var{ts} stores the value of \var{S.timestamp} at the linearization point of $V$.
\end{observation}

The following key lemma asserts that version lists are properly sorted.

\begin{lemma}
\label{lem:vsclp}
  The modifying \vcas\ operations are linearized in the order they insert \vnode{}s into the version list.
\end{lemma}

\begin{proof}
  Consider any two consecutive \vnode s $n_1$ and $n_2$ in the version list, where $n_1$ is inserted into the list before $n_2$, and let $V_1$ and $V_2$ be the \vcas\ operations that inserted $n_1$ and $n_2$ to the list, respectively.
  Recall that the linearization point of a modifying \vcas{} is at the read of the timestamp (Line~\ref{line:readTS}) of the \var{initTS} call that validates the timestamp on the \vnode\ that this \vcas{} appended to the version list.
  In particular, a modifying \vcas{} is linearized after it inserts its \vnode\ into the list (since  \var{initTS} cannot be called on a \vnode\ before it is inserted, by Lemma \ref{lem:initTS}(\ref{lem:head-before-initTS})), but before its \vnode\ is assigned a valid timestamp on Line~\ref{line:casTS} of \var{initTS}. By Lemma~\ref{lem:headInvalid}, a \vnode\ is assigned a valid timestamp before it is replaced as the head of the version list.
  That is, $V_1$ must be linearized before $n_1$'s timestamp was valid, and $n_1$'s timestamp became valid before $n_2$ was added to the list.
  Furthermore, $V_2$ was linearized after $n_2$ was added to the list.
  Therefore, $V_1$ is linearized before $V_2$.
\end{proof}

Now, we prove our main theorem which says that our \vCAS{} and \snapshot{} algorithms are linearizable and have the desired time bounds.

\begin{proof}[Proof (Theorem \ref{thm:vcas})]

 We show that the return values of each operation is correct with respect to their linearization points.  For \vcas\ and \vrd\ operations, this follows from Lemma \ref{lem:cur-val}.

  We prove this for \takess{} and \readss{} simultaneously.
  Suppose a $S$.\takess{} operation $T$ returns a timestamp $t$, which is passed into a $O$.\readss{} operation $R$.
  We  show that $R$ returns the value of $O$ at the linearization point of $T$.
  Let $h$ be the value of \VHead{} on line \ref{line:readss-head} of $R$.
  The timestamp of $h$ is valid after line \ref{line:readss-initTS} of $R$, and by Lemma \ref{lem:headInvalid}, the timestamps of all the nodes in the version list starting from $h$ are valid.
  This means that on line \ref{line:read-while}, \var{node->ts} is never \TBD{}.
  Let $n$ be the value of \var{node} at the last line of $R$ and let $V$ be the modifying \vcas{} operation that appended $n$.
  We know that $n$ is the first node in the version list starting from $h$ with timestamp less than or equal to $t$.
  Since $T$ is linearized when $S$.\var{timestamp} gets incremented from $t$ to $t+1$, by Observation \ref{obs:ts-valid}, $V$ is linearized before the linearization point of $T$.
  Since $R$ returns the value written by $V$, it suffices to show that no modifying \vcas{} operation gets linearized between the linearization points of $V$ and $T$.
  By Lemma \ref{lem:vsclp}, modifying \vcas{} operations are linearized in the order they appended \vnode s to the version list.
  Therefore, for all nodes that are older than $n$ in the version list, their modifying \vcas{} operations are linearized before the linearization point of $V$.
  Next, we show that all nodes in the version list that are newer than $n$ are linearized after $T$.
  From the while loop on line \ref{line:read-while}, we can see that all nodes that lie between $h$ and $n$ (including $h$, excluding $n$) have timestamps are larger than $t$.
  All nodes in the version list that are newer than $h$ also have timestamp larger than $t$ because they are appended after line \ref{line:readss-head} of $R$ and \var{S.timestamp} is already greater than $t$ at this step.
  Therefore, by Observation \ref{obs:ts-valid}, all nodes in the version list newer than $n$ are linearized after the linearization point of $T$.
  This means $V$ is the last modifying \vcas{} operation to be linearized before the linearization point of $T$, as required.

  The bounds on the step complexity of the operations can be derived trivially by inspection of the pseudocode.
\end{proof}


\section{Adding Linearizable Queries to Concurrent Data Structures}
\label{app:query}

In this section, we show how to use \vCAS\ objects to extend a large class of concurrent data structures that are implemented using reads and \cas\ primitives to support linearizable wait-free queries. Throughout the section we also use \NBBST{} as an example to show how our approach works. The idea of this construction is to replace \CAS\ objects with their versioned counterparts, and to use this to obtain \emph{snapshots} of the concurrent data structure. We can then run queries on the obtained snapshot, without worrying about concurrent updates to the data structure.

The techniques in this section are general. For many data structures, they allow translating any read-only operation on a sequential data structure into a linearizable query on the corresponding concurrent data structure. To achieve this generality, the techniques go through multiple layers of abstraction.
To make it more concrete, we show
examples of how to add specific linearizable queries to the Michael and Scott queue \cite{MS96} and \NBBST{} tree \cite{EFRB10} in Appendix~\ref{sec:examples}.

We present this construction in two parts. First, we define the concept of \emph{\sololy{} linearizable} queries.
A \emph{query} operation is an operation that does not modify the shared state (i.e., it is read-only).
A \emph{\sololy{} linearizable query} (or \emph{\solo{} query}) is a query operation that is only guaranteed to be correct if it is run \sololy.
Intuitively, a \solo{} query is one that can run on a ``snapshot'' of a concurrent data structure, and
it never changes the current state of the data structure.
They may be invoked while other operations are pending, but once invoked, they need to run to completion without any other process taking steps during their interval.
Section \ref{sec:making-solo-linearizable} describes how to transform a concurrent data structure that supports
\sololy\ queries (which we refer to as a \emph{\sololy{} linearizable} data structure) into a fully linearizable one using our \vCAS{} objects.
However, most concurrent data structures in the literature do not come with \sololy{} queries. 
In Section \ref{sec:add-solo}, we discuss how to add \sololy{} query operations to a given linearizable data structure.

\begin{definition}
	We denote by $\mathcal{H}(D,Q)$ the set of histories of concurrent data structure $D$ in which every operation instance from some set of query operations $Q$ is run \sololy.
\end{definition}

\begin{definition}
	A concurrent data structure $D$ is \emph{linearizable with \solo{} queries $Q$} if every history $H\in\mathcal{H}(D,Q)$ is linearizable. A query $q \in Q$ is called a \emph{\sololy{} linearizable} query on $D$. With clear context of \solo{} queries $Q$, we call $D$ a \emph{\sololy{} linearizable data structure}.
\end{definition}
%

The running time of a \solo{} query may depend on the concurrent state at which it is run. We denote by $T(q, C)$ the running time (number of steps) of a \solo{} query $q$ at concurrent state $C$ of the data structure.

\subsection{Making \Solo{} Queries Fully Linearizable}
\label{sec:making-solo-linearizable}
We now show how to transform a solo linearizable data structure $D$, implemented with \CAS\ objects, that has a set of \solo\ queries $Q_D$, into a fully linearizable data structure $D_\ell$. Let $L_D$ be the operations of $D$ that are not in $Q_D$. 
The transformation
uses our \vCAS{} objects in place of the regular \CAS{} objects of $D$.
It
preserves all existing
correctness guarantees (e.g., linearizability, strong linearizability,
sequential consistency) and progress guarantees (e.g.,
wait-freedom, lock-freedom) of the operations in $L_D$. Furthermore, it preserves the running time of operations of $L_D$ up to constant factors.
The time complexity of a linearizable query $q \in Q_D$ in $D_\ell$  is bounded by $q$'s time complexity in $D$, plus a contention term.

\begin{construction}\label{con:sololin}
	To obtain $D_\ell$ we replace every \CAS\ object with a \vCAS\ object, initialized with the same value.  All \vCAS{} objects are associated with a single \snapshot{} object.
	Each \cas\ or \rd\ 
	by an operation in $L_D$ on a \CAS\ object is replaced by a \vcas\ or \vrd\ (respectively) on the corresponding \vCAS\ object.
	%
	To perform a \solo{} query operation $q \in Q_D$ in $D_\ell$, a process $p$ first executes \takess\ on the \snapshot{} object, to obtain a handle~$h$.
	Then, for any \CAS\ object in $D$ that $q$ would have accessed, $p$ performs \readss($h$) on the corresponding \vCAS\ object. Recall that all operations in $Q_D$ are read-only, and thus never  perform a \cas.
\end{construction}

For this construction to be legal, we must show that the precondition for \readss($h$) holds. Namely, we need to show that \readss($h$) is never called on a \vCAS\ object that was created after the handle $h$ was produced. Intuitively, this is satisfied since no \vCAS\ object that was created after $h$ can be reachable in the data structure through version nodes with timestamp $h$ or earlier, which are the only version nodes that a query reads. The following claim makes the argument more formal.

\begin{claim}
	In Construction~\ref{con:sololin}, no \vCAS\ object $O$ is ever accessed using a \readss($h$) operation where $h$ was produced before $C$ was created.
\end{claim}

\begin{proof}
	Consider a query operation $q\in Q_D$ that uses handle $h$ to run on a data structure $D_\ell$ as prescribed by Construction~\ref{con:sololin}. We say that a \vCAS\ object is \emph{new} if it was created after $h$ was produced, and \emph{old} otherwise.
	Assume by contradiction that $q$ accesses a new \vcas{} object $O_{new}$. $O_{new}$ must be reachable from the root of $D_\ell$ for $q$ to access it. Note that by the way $D_\ell$ is initialized, the root must be an old \vCAS\ object. Without loss of generality, assume $O_{new}$ is the first new object that $q$ accesses in its execution. $O_{new}$ must be pointed to by some old \vCAS\ object $O_{old}$, through which $q$ accessed $O_{new}$. Since the only updates to \vCAS\ objects are via \vcas\ operations, $O_{old}$ must have been updated with a \vcas\ to point to $O_{new}$, thereby creating a new version of $O_{old}$. Note that since $O_{new}$ was created after $h$, this update must have also happened after $h$ was produced, and therefore the version of $O_{old}$ that points to $O_{new}$ has a timestamp larger than $h$. So, $q$ executing \readss($h$) on $O_{old}$ would not access the version pointing to $O_{new}$, but some older version instead. This contradicts the fact that $q$ reaches $O_{new}$.
\end{proof}

Using Construction~\ref{con:sololin}, we can make \solo\ queries linearizable with the bounds specified in the following theorem.




\begin{theorem}
\label{thm:sololin}
	Given a concurrent data structure $D$ with a set of linearizable operations $L_D$ and a set of \sololy{} query operations $Q_D$.
Construction~\ref{con:sololin} produces a linearizable data structure $D_\ell$ that supports operations from both $L_D$ and $Q_D$. This construction maintains the following properties:
	\begin{itemize}
		\item Operations from $L_D$ have the same progress properties in $D_\ell$ as in $D$, and their runtimes are increased by only a constant factor.
		\item Each operation $q \in Q_D$ costs $O(T(q, C_\ell) + W\cdot A)$ where $C_\ell$ is the concurrent state at which $q$ executes the \takess{} operation, $W$ is the number of \vcas\  operations concurrent with $q$ on memory locations accessed by $q$ in the execution, and $A$ is the maximum number of repeated accesses to the same object by the query.
	\end{itemize}
\end{theorem}

The proof is in Appendix~\ref{sec:soloproof}. We note that in most cases, the number of accesses that a query executes to the same object is 1 (or a small constant). If not, this bound can be improved by caching the values read from the data structure locally to avoid the extra overhead of reading it repeatedly from the concurrent data structure.

\subsection{Adding \Solo{} Queries to Linearizable Data Structures}
\label{sec:add-solo}

Concurrent data structures in the literature are usually designed to support a set of operations that are all linearizable. Thus, the question of whether \sololy{} linearizable operations can be easily incorporated is generally not considered when designing these data structures.
Is it always possible to run queries in a linearizable manner on a snapshot of any given data structure? How efficient can such queries be?
In this section, we address these questions.

While designing queries to run \sololy\ is certainly much simpler than designing them to be linearizable in the concurrent setting, it is still not as easy as designing queries for a sequential data structure.
This is because, in some cases, linearization points cannot be uniquely determined from the state of shared memory; instead, the linearization points may only be determined at the end of the execution, since they can depend on future events.
If this is the case, a query that is run \sololy\ cannot determine whether a pending update operation  has linearized or not, and, since the query may not change the state, it cannot enforce a placement of the linearization point.
Herlihy and Wing~\cite{herlihy1990linearizability} describe a queue implementation in which the linearization order of the enqueue operations depends on future dequeue operations. For that algorithm, no \solo\ query is possible. Herlihy and Wing~\cite{herlihy1990linearizability} point out that the difficulty in this scenario is the absence of an \emph{abstraction function} from states of the implementation to states of the abstract data type being implemented. 
We therefore define the notion of \emph{direct linearizability}, which intuitively means that there is always such a mapping from every concurrent state in an execution of the concurrent data structure to the abstract state of the abstract data type being implemented.

\begin{definition}
\label{def:abs}
An \emph{abstraction function} of a \solo{} linearizable data structure $D$ with solo queries $Q$ that implements an abstract data type $A$,
is a function $F: C_D \rightarrow C_A$ from concurrent states of $D$ to abstract states of $A$ such that for every history $H\in\mathcal{H}(D, Q)$, there exists a linearization of $E$ such that:
	\begin{enumerate}
		\item $F$ maps the initial state of $D$ to the initial state of $A$ (i.e., $F(C_D^{init}) = C_A^{init}$).
		\item If a concurrent state of $D$, $C_{D, 2}$ can be obtained from another concurrent state $C_{D,1}$ in $H$ without the linearization of any operation between $C_{D,1}$ and $C_{D,2}$, then they map to the same abstract state (i.e., $F(C_{D,1}) = F(C_{D,2})$).
		\item If a concurrent state of $D$, $C_{D, 2}$ can be obtained from another concurrent state $C_{D,1}$ in $H$ where operations, $op_1, \ldots, op_k$, linearized between $C_{D,1}$ and $C_{D,2}$ in this order, then $F(C_{D,2})$ is the state \Eric{a state? [to allow for non-deterministic types]} of $A$ that is obtained from applying $op_1 \ldots op_k$ in this order to $F(C_{D,1})$.
	\end{enumerate}
\end{definition}

Intuitively, the abstraction function respects the linearization points in the execution of $D$.
At first glance, it seems like the abstraction function's behavior is determined solely by the update operations from $D$.
However, query operations do have an indirect impact because they can affect the linearization points of the update operations, which affects the behavior of the abstraction function.
When the definition is applied to fully linearizable data structures, $Q = \emptyset$, so $\mathcal{H}(D,Q)$ is the set of all histories of $D$.

\begin{definition}\label{def:direct}
	A linearizable data structure is said to be \emph{directly linearizable} if it has an abstraction function.
\end{definition}

Direct linearizability is reminiscent of \emph{strong linearizability}~\cite{GHW11}. Strong
linearizability requires that the linearizations can be chosen for histories in a
\emph{prefix-preserving} way: for a prefix $H_p$ of a history $H$, the linearization of $H_p$
must be a prefix of the linearization of $H$. Thus, future events cannot determine
whether a given step in the execution was a linearization point or not. Intuitively, direct
linearizability requires that update operations be strongly linearizable, but does not require
the same behavior from query operations (that do not change the high-level state).
Furthermore, while strong linearizability only requires this ``prefix preserving'' behavior for
parts of the state that can be observed by operations of the data structure, direct
linearizability imposes this behavior on the entire shared state, regardless of the interface
through which operations of the data structure can access it. Appendix~\ref{sec:strong} shows
that strong and direct linearizability are incomparable. However, all strongly linearizable
data structures 
that we are aware of are also directly linearizable.


Consider the \NBBST{} binary search tree~\cite{EFRB10} outlined in Section~\ref{sec:query}. Recall that \NBBST{} implements the \emph{ordered set} abstract data type, with keys as elements.  To avoid special cases, the tree includes two leaves containing dummy keys.

\begin{proposition}
\label{prop:bst-abstract}
	Consider the function $F: C_E \rightarrow C_A$ that maps concurrent states of the \NBBST{} BST to states of the ordered set abstract data type as follows. Given a concurrent state $C_E$ of \NBBST{}, $F(C_E)$ is the set of keys in leaf nodes reachable from the root in $C_E$ except for the two dummy keys. $F$ is an abstraction function of \NBBST.
\end{proposition}

\begin{proof}
\Eric{The proof as previously written wasn't strictly correct, since it didn't handle the case of inserts and deletes that return false, so I rewrote it.  Is it too brief now?}
This theorem is proved as Lemmas 29 and 30 in the technical report \cite{EFRB10},
so we just sketch it here.
Initially, the tree has only the two leaves containing the dummy keys, which $F$
maps to the empty set, as desired.
Each \Insertop($k$) that modifies the tree is linearized at the child \CAS\ that adds a leaf
containing $k$ to the tree (and it is shown that $k$ was not present in the tree before
this change).
Similarly, each \Deleteop($k$) that modifies the tree is linearized at the child \CAS\ that removes a leaf
containing $k$ from the tree.
Each \Insertop($k$) that returns \false\ is linearized when $k$ is in a leaf of the tree, and each \Deleteop($k$) that returns \false\ is linearized  when there is no leaf
containing $k$, so these operations have no effect on the tree or on the abstract state of the set.
Each terminating \Findop($k$) returns \true\ if and only if $k$ appears in some leaf
at the linearization point of the \Findop.
It follows that each operation is linearized so that its effect on the set of keys stored in leaves
exactly matches its effect on the abstract state of the ordered set that the tree implements.
\end{proof}

Abstraction functions can help us both design \solo{} queries and prove their correctness.
It is often helpful to reason about \solo{} queries based on how they behave on each concurrent state.
For this purpose, we present the following definition.


\begin{definition}
	Let $op$ be an operation from a concurrent data structure $D$ and let $C$ be a reachable concurrent state, we define $op(C)$ to be the response value of $op$ when run solo on concurrent state $C$.
\end{definition}

Now we present a proof technique for showing that a read-only operation $q$ is a \solo\ query.
Consider a concurrent data structure $D$ that implements an abstract data type $A$ and is linearizable with \solo\ queries $Q$.
Suppose $D$ has an abstraction function $F$.
We add a query operation $q_A$ to $A$ to get the ADT $A'$ and we add $q$ to $D$ to get $D'$.
Our goal is to show that $D'$ is an implementation of $A'$ that is linearizable with \solo\ queries $Q \cup \{q_A\}$.
The following observation says that it suffices to show $q(C) = q_A(F(C))$ for all reachable concurrent states $C$.



\begin{observation}
\label{obs:proof-tech}
	If $q(C) = q_A(F(C))$ for any reachable concurrent state $C$, then $D'$ is an implementation of $A'$ that is linearizable with \solo\ queries $Q \cup \{q_A\}$. Furthermore, in this case, $F$ is still an abstraction function for $D'$.
\end{observation}

Note that the set of reachable concurrent states for $D$ does not change when we add a read-only operation $q_A$ to $D$.
The fact that $F$ is still an abstraction function for $D'$ is important because it allows us add \solo\ queries one at a time. This is summarized by the following observation.
\Eric{It's okay to leave this discussion, but I'm still not completely convinced that adding
queries one at a time is that big a deal.  If there is an abstraction function
wrt the update ops, why not just add a batch of queries all at once?}

\begin{observation}
\label{obs:combine-ops}
	Suppose two data structures have the exact same linearizable operations $L$, but different solo queries $Q_1$ and $Q_2$.
	If the same abstraction function works for both queries, then adding $Q_2$ to the first data structure results in a new data structure that is linearizable with solo queries $Q_1 \cup Q_2$.
\end{observation}

Next, we show how to use the abstraction function as a guide for designing solo queries.
If the abstraction function is computable and there is some way of viewing/traversing the state of shared memory, then an easy, but not necessarily efficient, method would be to first traverse the state of $D$, then use the abstraction function to arrive at an abstract state, and finally compute the query on the abstract state.
This query literally computes $q_A(F(C))$, so we can apply Observation \ref{obs:proof-tech}.
This is inefficient, since traversing the entire concurrent state often takes much longer than executing the query.
We show examples of how to compute queries designed for a sequential version of the data structure on a concurrent state.

\subsubsection{Solo queries for \NBBST{}}
\label{sec:solo-bst}

\Eric{Changed external to leaf-oriented for consistency with other sections}
Consider the \NBBST{}.
The concurrent state of the \NBBST\ includes a lot of information used to coordinate concurrent updates.
By removing everything except the root pointer, the \var{key}, \var{left}, \var{right} fields of each \var{Internal} node, and the \var{key} fields of each \var{Leaf} node, we end up with a standard leaf-oriented BST (with child pointers, but no parent pointers).
This means that sequential read-only queries that work on a leaf-oriented BST, such as predecessor or range queries, can be run on the \NBBST\ as is, because they only access fields that we keep.
In the following theorem, we show that these read-only queries can be added to the \NBBST\ as \solo\ queries without any modification.


\begin{theorem}
\label{thm:bst-main}
	Let $Q$ be a set of read-only, sequential operations on a leaf-oriented BST implementing a set of abstract queries $Q_A$.
	Let $A'$ be an ADT that supports ordered set operations as well as queries from $Q_A$.	
	Adding the operations in $Q$, without modification, to the \NBBST\ yields a concurrent implementation of $A'$ that is linearizable with \solo{} queries $Q$.
\end{theorem}

\begin{proof}

  Let $F$ be the abstraction function from Proposition \ref{prop:bst-abstract} for \NBBST.
  Pick any $q \in Q$ and let $q_A \in Q_A$ be the abstract operation that it implements.
  Our goal is to show that $q(C) = q_A(F(C))$ for all reachable concurrent states $C$. Then we can apply Observations \ref{obs:proof-tech} and \ref{obs:combine-ops} to complete the proof.

  We begin by defining a mapping $F_A$ from states of leaf-oriented BSTs to states of $A$ and a mapping $F_I$ from concurrent states to states of a leaf-oriented BST.
  (The $I$ in $F_I$ stands for intermediate state because it is in between the abstract state and the concurrent state.)
  For $F_A$ we use the textbook mapping which maps an leaf-oriented BST to the set of keys that appear in its leaves.
  Given an leaf-oriented BST state $S$, return value of $q$ on state $S$ (denoted $q(S)$) equals $q_A(F_A(S))$.
  To compute the mapping $F_I$, we start with a concurrent state and remove everything except the root pointer, the \var{key}, \var{left}, \var{right} fields of each \var{Internal} node, and the \var{key} fields of each \var{Leaf} node.
  Since $q$ only accesses the fields that we keep, it cannot tell the difference between running on a concurrent state $C$ and running on $F_I(C)$.
  Therefore $q(C) = q(F_I(C))$.
  It is easy to verify that $F = F_A \circ F_I$ and this completes the proof because $q(C) = q(F_I(C)) = q_A(F_A(F_I(C))) = q_A(F(C))$.
\end{proof}

We apply Construction~\ref{con:sololin} on top of Theorem \ref{thm:bst-main} to get a data structure $D_\ell$ that supports \var{insert}, \var{delete}, and \var{find}, as well as linearizable implementations of any query for which there is a read-only sequential algorithm.
By Theorem \ref{thm:sololin}, we maintain the efficiency of \var{insert}, \var{delete}, and \var{find} up to constant factors and for each new query operations in $D_\ell$, it is wait-free and its runtime is proportional to the sequential cost of the query plus $W\cdot A$, where $W$ is the number of \vcas{} operations that occur during the query and that operate on objects accessed by the query, and $A$ is the maximum number of repeated accesses to the same object by the query.
Most read-only, sequential operations on a leaf-oriented BST, such as \var{predecessor} and \var{range\_query}, can be written so that each query accesses a memory location no more than a constant number of times.
For such operations, the added cost is just $O(W)$.
For example, consider a \var{range\_query} operation that computes the list of keys within a certain range.
If we start with a sequential implementation that takes $O(h + k)$ time, where $h$ is the height of the BST and $k$ is the number of keys within the specified range, then the corresponding concurrent query in $D_\ell$ would take $O(h' + k + W)$ time, where $h'$ is the height of the concurrent tree at the linearization point of the operation.

The \NBBST\ is an easy example because the function $F_I$ from the proof of Theorem \ref{thm:bst-main} is essentially an identity function.
We show a more complicated example with Harris's linked list in Appendix \ref{sec:solo-harris-ll}.
With an appropriately defined $F_I$, the proof structure we used for Thereom \ref{thm:bst-main} works for Michael and Scott's queue \cite{MS96}, Harris's linked list \cite{Harris01}, Natarajan and Mittal's BST \cite{NM14}, and chromatic BSTs \cite{BER14}.
For these algorithms, we need to slightly modify the sequential, read-only operations to make them solo queries.
The mapping $F_A$ is always defined to be the standard mapping from sequential to abstract state and the key property to prove is that $F = F_A \circ F_I$.

\section{Proof of Theorem~\ref{thm:sololin}}\label{sec:soloproof}

\Eric{Move this to where the theorem is now}

\begin{proof}
	We construct $D_\ell$ from $D$ as described by Construction~\ref{con:sololin}.
	We want to show that $D_\ell$ is a linearizable implementation of the data structure $D$ with all of its operations in $Q_D$ and $L_D$. We do so by mapping each history $H_\ell$ of $D_\ell$ to a history $H$ of $D$ in which all \sololy{} linearizable queries are run in isolation, and in which all read and \CAS\ operations return the same values as the \readss{}, \vrd{} and \vcas{} operations in $H_\ell$.
	Furthermore, $H_\ell$ and $H$ will have the exact same high-level history.
	
	
	Given a history $H_\ell$ of $D_\ell$, we map it to a history $H$ of $D$ as follows.
	(1) For every query operation $q \in Q_D$ we move all \readss($h$) operations executed by $q$ to appear immediately after the \takess\ that returned $h$, in the same order. We then remove the \takess{}, and replace all \readss($h$) operations with reads of the corresponding \CAS\ objects of $D$.
	(2) For every \vrd{} or \vcas{} operation that appears in $H_\ell$, we simply map it to a \rd\ or \cas\ (respectively) on the corresponding \CAS\ object in $D$, without moving it in the history.
	
	Note that by the definition of the \vCAS\ object, \vrd{} and \vcas{} behave the same as \rd\ and \cas\ in \CAS\ objects. Furthermore, note that the only operations that are moved in $H_\ell$ to form $H$ are \readss\ operations, which do not affect the state of the \vCAS\ object they operate on. Thus, all \rd\ and \cas\ operations in $H$ return the same values that their corresponding \vrd{} and \vcas\ operations returned in $H_\ell$.
	Recall from the definition of the \vCAS\ object that \readss($h$) always returns the value of the \vCAS\ object it operates on at the time that handle $h$ was produced. Thus, a \rd{} of a \CAS\ object in $H$ executed at the concurrent state at which the handle $h$ was produced returns the same value as the \readss($h$) operation anywhere in the history $H_\ell$. 
	
	Since $D$ is a \sololy{} linearizable data structure and $H$ is a legal history of $D$ in which all \solo{} operations run in isolation, $H$ is a linearizable history. Since all operations in $H$ return the same values as they return in $H_\ell$, $H_\ell$ is also linearizable. Furthermore, we can linearize any operation $\ell \in L_D$ in $H_\ell$ at the same step as it linearizes in $H$ (where we map steps of $H_\ell$ to steps of $H$ in the same way as above), and any operation $q \in Q_D$ at its \takess{} operation.
	
	To show the required running time bounds, note that operations in $D$ and operations in $D_\ell$ access the same number of base objects. The difference in running time for operations in $L_D$ is strictly due to the time it takes to access the \vCAS\ object for \vcas\ and \vrd{} operations. By Theorem~\ref{thm:vcas}, this amounts to constant overhead. Operations $q\in Q_D$ execute one \takess{} operation, which takes constant time, and then replace every \rd\ they would do in the implementation of $D$ with a \readss($h$) where $h$ is the handle returned by the \takess.  By Theorem~\ref{thm:vcas}, the running time of each \readss($h$) is proportional to the number of successful \vcas\ operations on that object since $h$ was produced. Note that all such \vcas\ operations on all \vCAS\ objects that the query accesses are concurrent with the query itself. Thus, we get our desired time bounds.
\end{proof} 

\section{Relationship of Strong Linearizability to Direct Linearizability}
\label{sec:strong}

Strong linearizability was introduced by Golab, Higham and Woelfel \cite{GHW11}
to provide a stronger guarantee that permits reasoning about concurrent
executions that involve randomness.
At first glance, it seems that this condition might be what is required
for our \approach\ to be applicable.
We show here that strong linearizability is \emph{not} comparable to direct linearizability
(defined in Definition \ref{def:direct}),
which is the property required for our \approach{}.

Intuitively, an implementation is strongly linearizable if linearization points for
each operation can be chosen as the execution proceeds, without needing 
to know what happens later in the execution.
(See \cite{GHW11} for the formal definition; the informal definition will suffice
for this discussion.)

We first show that strong linearizability does not imply direct linearizability.
Consider a non-deterministic ADT $A$ stores a single bit and provides the following two operations. 
\op{Write-random-bit}, which sets the bit to either 0 or 1, non-deterministically, 
and returns {\it ack}).
\op{Read} simply returns the current value of the bit.
Let $D$ be an implementation that delays the choice of the random bit written
by a \op{write-random-bit} until the first subsequent \op{read} operation.
More precisely, $D$ uses a writable \CAS\ object $X$ with three possible states, $\bot, 0,1$.
A \op{write-random-bit} operation simply writes $\bot$ into $X$.
A \op{read} does a \cas($X,\bot,random(0,1)$) and returns the new value if the \cas\ is successful, or the old value if the \cas\ is unsuccessful.
Since each operation performs only one shared-memory access in $D$,
that access must serve as the linearization point of the operation.
It is easy to see that this linearization is correct.
Thus, linearization points can be determined without having to know what happens
later in the execution.  In other words, $D$ is strongly linearizable.
However, $D$ is not directly linearizable:
when $X$ is in state $\bot$, there is no abstract state that can
be used as the value of the abstraction function $F$.
If $F(\bot)=0$, then an execution in which the subsequent read returns 1
would violate Definition \ref{def:abs}.  A similar problem arises if $F(\bot)=1$.

Indeed, our \approach\ would fail if we tried to apply it to $D$:  
if a \takess\ is performed between 
a \op{write-random-bit} and the first subsequent \op{read},
reading the snapshot would yield $\bot$, and it would be impossible to conclude
what state of $A$ this corresponds to.
Thus, strong linearizability is not a sufficient condition
for our \approach\ to be applicable.

Next, we show that direct linearizability does not imply strong linearizability.
The snapshot ADT \cite{AADGMS93} stores a vector of values and allows processes to update components
of the vector or perform a \op{scan} that reads the whole vector atomically.  The classic implementation of \cite{AADGMS93}
implements a snapshot using an array of values (with associated timestamps to avoid ABA problems).
Updates are performed by writing the new value to the appropriate location in the array
and changing its timestamp.  A \op{scan} reads the array repeatedly until 
getting identical results twice.  This implementation is not strongly 
linearizable \cite{GHW11}, but it is directly linearizable:  the abstraction function simply
strips the timestamps from the elements stored in the array to get the state of the
ADT.


\section{Examples}\label{sec:examples}

Our first example focuses on the concurrent queue implementation by Michael and Scott presented in~\cite{MS96}.
We will call this implementation \MSQ. Thus, we have $D = \MSQ$, and $A$ is the abstract data type of a FIFO queue
that stores integers and supports the operations \enqueue\ and \dequeue.

We start by describing how \MSQ\ works.
\MSQ~\cite{MS96} implements the queue using a simply-linked list of Node objects, each
storing a key and a \var{next} pointer pointing to the next Node.
Two pointers, called \var{Head} and \var{Tail}, point to the
first and the last element of the list that implements the queue, respectively.
The first Node of the list is always a dummy Node.
Thus, the elements of the queue are the keys of the Nodes starting from the
second Node of the list up until its last Node.
Initially, the list contains just the dummy Node, whose key can be arbitrary
and its \var{next} pointer is equal to NULL. At each point in time, the list
contains those elements that have been inserted in the queue
and have not yet been deleted, in the order of insertion.
It also contains the last element that has been dequeued
as the first element of the list (i.e., as the dummy Node).

To insert a key $k$ in the queue, a process $q$ has to call \Enqueue($k$).
\Enqueue\ first allocates a new Node $nd$ with key $k$ and its \var{next} field equal to NULL.
It then reads \var{Tail} and checks whether the \var{next} field of the Node it
points to is equal to NULL. If this is so, \var{Tail} points to the last
element of the queue, and \Enqueue\ attempts to insert $nd$
after this Node using a \cas. If this \cas\ is successful,
then $q$ performs one more \cas\ trying to update \var{Tail} to point to $nd$.
Otherwise, some other process managed to
insert its own Node as the next to the last one, so $q$ has to retry.
If the Node pointed to by \var{Tail} does not have its \var{next} field equal to NULL,
then some process has managed to insert its own Node as the next Node
to the one pointed to by \var{Tail} but it has not yet updated \var{Tail}
to point to this Node (i.e., \var{Tail} is falling behind). To ensure lock-freedom, whenever $q$
discovers that \var{Tail} is falling behind, it helps
by updating \var{Tail} to point to the last Node of the list, before it restarts its own operation.

A process $q$ executing \Dequeue, reads both \var{Head} and \var{Tail}. If they both point to the same Node and the \var{next} field
of this Node is NULL, then the queue is empty (it contains just the dummy Node)
so \false\ is returned. If they point to the same Node, but the \var{next} field of this Node
is not NULL, then \var{Tail} is falling behind, so $q$ has to help by performing a \cas\ to update
\var{Tail} to point to the last Node of the queue before it retries its own operation.
If \var{Head} and \var{Tail} do not point to the same Node, \Dequeue\ reads the key of the
second Node of the list and performs a \cas\ in an effort to update \var{Head} to point
to this Node. If the \cas\ is successful,
\Dequeue\ completes by returning the key that it read (and the Node from where it read this key becomes the dummy Node).
Otherwise, $q$ restarts the execution of \Dequeue.

In \MSQ,
the \var{Head} pointer always points to the first element of the list,
whereas the \var{Tail} pointer always points either to the last
or to the second last pointer of the list. This implies that whenever
the \var{next} pointer of the last element of the list changes to
point to a newly inserted Node, \var{Tail} points to the last
Node of the list.
Moreover, whenever \var{Head} is updated, \var{Tail} does not point to the
first element of the list.
These properties and the way helping is performed make it possible
to assign linearization points to the queue operations in two different ways.
A \Dequeue\ is linearized when \var{Head} is updated to point to the
list Node whose element the \Dequeue\ returns. An \Enqueue\ can be linearized
either at the point the \var{next} field of the last Node changes to
point to the newly inserted Node, or it can be linearized when the
\var{Tail} pointer changes to point to the newly inserted Node (notice
that the latter change might not be performed by the same process that initiated the \Enqueue).
Note that whenever \Dequeue\ interferes with \Enqueue,
i.e., whenever there is just one element in the queue, \Dequeue\ first
updates \var{Tail} to point to the last Node (if needed) and then performs
the deletion. In this way, \var{Head} is never ahead of \var{Tail} and
therefore the linearization point of an \Enqueue\ always precedes
the linearization point of the \Dequeue\ that deletes the element
that the \Enqueue\ inserted in the list. (Note that this is true for both ways of assigning linearization points.)
It is also not hard to prove that the list is always connected,  and the Nodes are appended at the end of the list,
and that they are extracted from the beginning of it,
in the order defined by the sequence of the linearization points assigned to \Enqueue\ operations.

Note that the way we choose to assign linearization points allows us to
determine the annotations for the executions of \MSQ\ (in a straightforward way).
Note that both linearization schemes, assign the linearization point
of an operation at the point in time that a concrete \cas\ is executed,
i.e. each linearization point is assigned at the point that an internal actions of the \MSQ\ I/O automaton
occurs. This allows us to
come up with an abstraction {\bf function} in each case. 

\remove{
A \Dequeue\ is linearized when \var{Head} is updated to point to the
list Node whose element the \Dequeue\ returns. Similarly, an \Enqueue\ is linearized
when the \var{Tail} pointer changes to point to the newly inserted Node (notice
that this might not happen by the same process that initiated the \Enqueue).
This assignement is correct since whenever \Dequeue\ interferes with \Enqueue,
i.e., whenever there is just one element in the queue, \Dequeue\ first
updates \var{Tail} to point to the last Node (if needed) and then performs
the deletion. In this way, \var{Head} is never ahead of \var{Tail} and
therefore the linearization point of an \Enqueue\ always precedes
the linearization point of the \Dequeue\ that deletes the element
that the \Enqueue\ inserted in the list. It is also not hard to prove
that the list is always connected,  and the Nodes are appended at the end of the list,
as well as they are extracted from the beginning of it,
in the order defined by the sequence of the linearization points assigned to \Enqueue\ operations.
}

\Y{
In the rest of this section, we choose the second way to assign linearization points (call this
linearization scheme $L_2$).
If $C_0$ is the initial configuration of \MSQ,
then $F(\hat{C}_0)$ is an initial state where the queue does not contain any element.
For every execution $\alpha$ and every configuration $C$ of $\alpha$,
$F(\hat{C})$ is the state of the queue that results when all \enqueue\ and \dequeue\
operations that have been linearized by $C$ occur starting from an empty queue
in the order determined by their linearization points.
}

Figures~\ref{fig:msqueue-query}  and~\ref{fig:msqueue-enqdeq}
show how to implement two kinds of read-only queries on top of \MSQ\ using \vCAS\ objects.
The first, called \peekHeadandTail, returns the values of the first and the last element
in the queue. The second implements \scan, i.e., it returns a set containing the keys of all queue Nodes.

To implement these queries, we have to perform the simple changes to \MSQ\ described in Section~\ref{sec:query}.
We call the resulting algorithm \VMSQ.
\Eric{I think the following sentence is confusing because we said \vCAS\ objects store a value.  This sentence is confusing the value that is stored with how the \vCAS\ is implemented.}
In \VMSQ, \var{Head} and \var{Tail} are \vCAS\ objects storing references to VNode objects whose \var{val} field
points to the first and the last Node of the list, respectively.
Similarly, the \var{next} field of each Node is a \vCAS\ object storing a reference to a VNode object
that  contains the pointer to the next Node in the queue.
The code for \Enqueue\ and \Dequeue\ remains unchanged but every read to
\var{Head}, \var{Tail} or to the \var{next} field of a Node has to be replaced with an invocation of \vrd\ (on the same object).
Similarly, every \cas\ on each of these objects, has to be replaced with a \vcas\ (on the same object
with the same old and new values). We remark that all \vCAS\ objects
are associated with
a single \snapshot\ object $S$.  A \takess\ is invoked on this \snapshot\ object at the
beginning of every query.

\PeekHeadandTail\ starts by 
performing a \takess\
and storing the resulting handle into a local variable \var{ts}.
Finally, it executes \readss(\var{ts})\ to read both \var{Head} and \var{Tail}
and returns the values it read.

A \scan\ first executes \takess\ to get a handle \var{ts}, and also reads \var{Head} and \var{Tail} using \readss(\var{ts}). Then, it executes a while loop to traverse the list
starting from the node pointed to by the value read in \var{Head} until the node pointed to by the value
read in \var{Tail}. It uses a set to collect the pointers of the nodes it traverses (other than the first one) and returns this set
at the end. On each node, it calls \readss(\var{ts})\ to move to the next node. This ensures that updates
that occured in the list after the point that the global timestamp was increased will not be included in the set.

\Y{
In \VMSQ, linearization points to \enqueue\ and \dequeue\ operations are assigned as in the $L_2$ scheme.
A query is linearized at the point that the \takess{} it invokes is linearized.
Note that the queries would be implemented differently in case the abstraction function was
determined based on $L_1$. Specifically, then both \peekHeadandTail\ as well as \scan, would have to
read the \var{next} field of the node pointed to by \var{Tail} and if its not NULL, adjust their responses
based on which is the actual last node in the list.
}

\remove{Figure~\ref{submitted figure by email}(a) shows the shared state maintained by \VMSQ,
after \enqueue(3) and \enqueue(10) have been applied onto an empty queue
(containing only the dummy node which has key value equal to $\bot$).
We have assumed that the initial value of the counter pointed to by $TS$
is 0 and that no query is performed until these two \enqueue\ operations
have been completed. Figure~\ref{submitted figure by email}(b) assumes that
then a \scan\ query starts by some process $q$, which reads 0 in the counter and sets the value of the counter to 1.
Then, \enqueue(10) and a \dequeue\ take place. We have assumed that the timestamp of all \vcas\
executed by these two operations is 1.
Note that the \scan\ uses the value 0 as its timestamp and therefore it will return $\{3,10\}$.
By the way we assign linearization points, we conclude that this response is correct.}

\remove{
\vspace*{.1cm}
\noindent
{\bf Optimizations.} We remark that \VMSQ\ would be correct, even if we do not implement
the \var{next} field of a Node as a \vCAS\ object, i.e., even if the \var{next} field
of each node is a simple pointer to the next Node and standard \cas\ is used to update it.
\Youla{Let me know, if you would like me to ellaborate more on this part.}
}


\remove{
\here{Youla: The way linearization points are assigned to the queue operations of the \MSQ\ algorithm
in some textbooks, e.g. in the book by Herlihy and Shavit, but also (and possibly more significantly)
in the original paper by Michael and Scott
is the following: "An enqueue takes effect when the allocated node is linked to the last node in the linkedlist.
A dequeue takes effect when Head swings to the next node." This assignment does not work in our scheme
due to the way we implement queries. If for a scan we start from Head and traverse the list
until we see a NULL pointer, I think it would work. However, by having scan reading Tail,
and given that Tail might be falling behind, we may miss the last node which could have been appended
with timestamp 0 without performing the update of Tail before the scan reads Tail's value.
However, if we implement peekEndPoints by traversing the entire list to discover which is the
last node, then it seems that some of the claimed bounds would not be correct. Finally, a query could
read Tail, check whether it is falling behind, and if it does either help or manually include the single
node that it may miss (in case this node is relevant in terms of its timestamp). However, that would require the programmer
to know all the details of the algorithm. }
}

\begin{figure*}[t]
\begin{minipage}[t]{.5\textwidth}
\begin{lstlisting}[linewidth=.99\columnwidth, numbers=left, frame=none]
// View both ends of the queue
<Value, Value> @\hl{peekEndPoints}@() {
  Node* HNode, TNode;
  int ts = TS.takeSnapshot();
  HNode = Head.readSnapshot(ts);
  TNode = Tail.readSnapshot(ts);
  if (HNode != TNode) {
	  HNode = HNode->next.readSnapshot(ts);
    return <HNode->val, TNode->val>; }
  return <@$\bot$@, @$\bot$@>;
}

List<Node*> @\hl{SCAN}@() {
  List<Node*> Result; 	\\ initially empty
  int ts = TS.takeSnapshot();
  Node* q = Head.readSnapshot(ts);
  Node* last = Tail.readSnapshot(ts);
  while(q != last) {
    q = q->next.readSnapshot(ts);
    Result.append(q); }
  return Result;
}
\end{lstlisting}
\caption{Two Example Query Operations for \VMSQ.}
\label{fig:msqueue-query}
\end{minipage}\hfill
\begin{minipage}[t]{.42\textwidth}
\StartLineAt{23}
\begin{lstlisting}[linewidth=1.1\textwidth, numbers=left, frame=none]
class Node {
  Value val;
  @\hl{VersionedCAS}@<Node*> next;
  Node(Value v, Node* n) :
   {val = v; next = n;}
};

class Queue {
  @\hl{VersionedCAS}@<Node*> Head, Tail;
  @\hl{Camera TS;}@

  /*
  @\hl{Same code except replacing each}@
  @\hl{read of Head, Tail or the next}@
  @\hl{field of a Node with readSnapshot}@
  @\hl{and replacing each CAS on these}@
  @\hl{variables with a vCAS}@
  */
  ...

};
\end{lstlisting}
\caption{Node representation and Code updates for \Enqueue\ and \Dequeue\ in \MSQ.}
\label{fig:msqueue-enqdeq}
\end{minipage}
\end{figure*}



\subsection{A Versioned Concurrent BST Implementation based on the NBBST}
\label{app:EFRBexample}

\Eric{It should be much briefer.}
We start with a brief, informal description of the concurrent
binary search tree (BST) implementation provided in~\cite{EFRB10}, which we
will call \NBBST.
\NBBST\ implements a leaf-oriented tree, i.e., a tree that represents a set
whose elements are the keys stored only in the leaf nodes of the tree.
The tree is full, i.e., every internal node has exactly two children.
Moreover, the tree satisfies the following sorting property:
for every internal node $v$ with key $K$, the key of every node in the left subtree of $v$
is smaller than $K$, whereas every node in the right subtree of $v$
has key larger than or equal to $K$.

\NBBST\ supports three operations, \Insertop($k$), \Deleteop($k$), and \Findop($k$),
where $k$ is a key.
All three operations start by calling \Search($k$), a routine that searches for $k$
by following the standard BST searching algorithm. \Search\
returns a pointer $l$ to the leaf that it arrives, a pointer $p$ to its parent node,
and a pointer $gp$ to the grandparent of this leaf.
\Findop\ simply checks whether the leaf node returned by \Search\ contains key $k$.
If it does, it returns \true, otherwise \false\ is returned.

In its sequential version, \Insertop\ replaces the leaf that the \Search\ arrives at
with a BST of three nodes, two leaves
containing the key of the node pointed to by $l$ and the newly-inserted key,
and an internal node containing the larger key among the keys of the two leaves.
The replacement
is performed by switching the appropriate pointer of $p$ from $l$ to the
root of this BST. \Deleteop\ essentially performs  the inverse action:
it switches the appropriate child pointer of $gp$ from $p$ to the sibling
of the node pointed to by $l$, thus replacing a part of the tree that is comprised of three
nodes (one of which is the leaf to be deleted) with one node, namely the sibling of the node
to be deleted.

To avoid synchronization problems, \NBBST\ uses \cas\ to apply
a change to a child pointer of a node. Such a \cas\ is called a
{\em child} \cas. Morover, it flags a node when its child pointer
is to be changed, and unflags it after the change of the
child pointer has been performed. These two types of \cas\
are called {\em flag} \cas, and {\em unflag} \cas, respectively.
Thus, \NBBST\ uses flagging to ``lock'' a node (in a non-blocking manner)
whose child pointer is to be changed.
\NBBST\ marks an internal node
when the node is to be deleted. It does so by executing a {\em mark} \cas.
A marked node remains marked forever.
To implement flagging and marking, each node has a two-bit status field, which can have
one of the following four values: \Clean, \flag\ for insertion, \flag\ for deletion,
or \Mark. A  flag or mark \cas\ can succeed only if it is applied on a node whose status is \Clean.
If a flag \cas\ fails, the process that performed the flag \cas\ retries the execution of its operation
by starting it from scratch. If the operation is a \Deleteop\
and the flag \cas\ succeeds but the mark \cas\ fails, then the process
first unflags $gp$, using a {\em backoff} \cas, and then retries the execution of its operation.

To ensure lock-freedom, each process executing an operation
records in an {\em Info} object all the information needed by other processes to complete the operation.
A pointer to such an object is stored together with the status field of a node (and they are manipulated
atomically). A process  $q$ that fails to flag or mark a node helps the operation
that has already flagged or marked the node to complete (by reading the necessary information in the Info object
pointed to by the status field of the node). Then, $q$ restarts its own operation. This ensures that a single operation
cannot repeatedly block another operation from making progress. 
Thus, lock-freedom is ensured.

For \NBBST, it is proved that each node a \Search($k$) visits was in the tree, on the search path for $k$,
at some time during the \Search. \Search($k$) is linearized at the point when the leaf it returns was on the search path for $k$.
The \insertop\ and \deleteop\ operations that return \false\ are linearized at the same point
as the \Search\ they perform. Every \insertop\ or \deleteop\ operation that returns \true\ has a unique successful child \cas\
and the operation is linearized at that child \cas.
\Eric{What about those that modify tree but die before returning?}

Figures~\ref{fig:vbst-insdel}  and~\ref{fig:vbst-queries}
show how  we can modify \NBBST\ to support queries using \vCAS\ objects.
We call the resulting algorithm \VBST.
They also provide pseudocode for
\RSum($a,b$), a query that
returns the set of those keys  in the implemented set
that are larger than or equal to $a$ and smaller than or equal to $b$.

\Eric{Same problem here as for queue above}
In \VBST, the \var{child[LEFT]} or \var{child[RIGHT]} field of each internal Node $v$ is a \vCAS\ object storing a reference to a VNode object
that  contains the pointer to the left or right child, respectively, of $v$ in the tree.
Moreover, every read of the \var{child[LEFT]} or \var{child[RIGHT]} field of a Node is replaced
with an invocation of \vrd\ (on the same object).
Similarly, every \cas\ on each of these fields, is replaced with a \vcas\ (on the same object
with the same old and new values). We remark that the \var{status} field of a Node does not have
to be a \vCAS\ object, as queries simply ignore the flag and mark signs on the Nodes.
Notice that all \vCAS\ objects
are paired with the same \snapshot\  on which a \takess\ is invoked
at the beginning of the execution of every query.

\RSum($a,b$)\ first performs a \takess.
Then, it calls the recursive function \RSTraverse\ which traverses
part of the tree
to perform the required calculation. Pseudocode for 
\RSum\ is provided in Figure~\ref{fig:vbst-queries}.

In \VBST, linearization points can be assigned to \insertop\ and \deleteop\ operations  the same way
as in \NBBST. A query is linearized at the linearization point of the \takess\  it invokes on Line~\ref{line:CAI}.

\begin{figure*}[t]
\begin{lstlisting}[linewidth=.99\textwidth, numbers=left, frame=none]
class InternalNode {		// subclass of Node
   Key key;
   HelpRec hr;
   array VersionedCas<Node *> child[LEFT,RIGHT]
   Node(Key k, HelpRec up, Node* lc, Node* rc) :
   {key = k; hr = up; child[LEFT] = lc; child[RIGHT] = rc;}
};

class BST {
	Node* Root;
  @\hl{Camera TS; }@

  // @\hl{Same code except replacing read of child[LEFT] or child[RIGHT] field}@
  // @\hl{of a Node with readSnapshot and CAS on them with vCAS}@
  // @\hl{appropriate initialization is necessary}@
  ...

};
\end{lstlisting}
\caption{Node representation and Code updates for \insertop\ and \deleteop\ in \NBBST.}
\label{fig:vbst-insdel}
\end{figure*}

\begin{figure*}[t]
\begin{lstlisting}[linewidth=.99\columnwidth, numbers=left, frame=none]
Integer @\hl{RangeSum}@(Key a, Key b) {
  int ts = TS.takeSnapshot();  @\label{line:CAI}@
  return RSTraverse(Root, ts, a, b); }

Integer RSTraverse(Node *nd, int ts, Key a, Key b) {
  if (nd points to a leaf node and nd->key is in [a,b]) return nd->key;
  else {
     if (a <= nd->key) {
	     return RSTraverse(nd->child[RIGHT].readSnapshot(ts), ts, a,b);  }
     else if (b < nd->key) {
	     return RSTraverse(nd->child[LEFT].readSnapshot(ts), ts, a, b); }
     else {
		  return RSTraverse(nd->child[LEFT].readSnapshot(ts), ts, a, b)  +
		      RSTraverse(nd->child[RIGHT].readSnapshot(ts), ts, a, b); }
  }
}
\end{lstlisting}
\caption{An Example Query Operation for \VBST.}
\label{fig:vbst-queries}
\end{figure*}

\remove{
Boolean @\hl{Find}@(Counter *TS, Key k) {
  Node *q;
  int ts = TS.read();
  TS.CompareAndIncrement();

  q = R;
  while (q points to an internal node) {
	  if (q->key > k)  q = q->child[LEFT].VRead(ts);
	  else  q = q->child[RIGHT].VRead(ts);
  }
  if (q->key == k)  return TRUE;
  return FALSE;
}
}


\section{Solo queries for Harris's Linked List}
\label{sec:solo-harris-ll}

Each node in the Harris linked list contains a \var{key} field and a \var{next} field.
The next field stores a pointer that is potentially marked, indicating that the node containing this next field has been logically deleted.
A state of the Harris linked list consists of a set of nodes and a pointer to the first node in the linked list.


The technique we use from constructing \solo{} is similar to the technique we applied to the \NBBST{} in Section \ref{sec:solo-bst}.
Just like the \NBBST{}, Harris's linked list implements an ordered set ADT.
First, we define $F_A$ to be the standard mapping from sequential linked lists to ordered sets which basically maps a linked list to the set of keys that appear in its nodes.
If $q_I$ is a read-only, sequential linked list operation implementing the abstract query $q_A$, then by the correctness of $q_I$, we know that $q_I(S) = q_A(F_A(S))$ for all sequential states $S$.
Just like in Section \ref{sec:solo-bst}, we use $q_I(S)$ to denote the return value of $q_I$ when run on state $S$.
Next, we define a mapping $F_I$ from concurrent states to sequential linked list states such that $F = F_A \circ F_I$ is an abstraction function.
For each sequential, read-only linked list operation $q_I$ implementing the abstract query $q_A$, we show how to modify $q_I$ into a read-only operation $q$ for the Harris linked list such that for all reachable concurrent states $C$, $q(C) = q_I(F_I(C))$.
Once we have these facts, it is fairly straight-forward to complete the proof.
From the equalities we've proven $q(C) = q_I(F_I(C)) = q_A(F_A(F_I(C))) = q_A(F(C))$ for all reachable concurrent states $C$, so by Observation \ref{obs:proof-tech}, $q$ can be added to Harris's linked list as a \solo{} query.

In Harris's original paper \cite{Harris01}, he linearizes successful \var{insert} operations when the node being inserted gets connected to the data structure and successful \var{delete} operations when the node being deleted gets marked for deletion (i.e. when the node gets logically deleted, not physically deleted).
We need to define $F_I$ so that the abstraction function $F_A \circ F_I$ is consistent when these linearization points.
To compute $F_I(C)$, we physically delete (i.e. unlink) all the logically deleted nodes from $C$ and return the resulting linked list as the sequential state.
Using an  argument similar to the proof of Proposition \ref{prop:bst-abstract}, we can show that $F = F_A \circ F_I$ is an abstraction function for Harris's linked list.

Now we show how to transform $q_I$ into $q$.
Let $q_I$ be a read-only, sequential linked list operation implementing the abstract query $q_A$.
To construct $q$, whenever $q_I$ reads the next pointer of a node, change it to call the \var{getNext} function implemented in Figure \ref{fig:getNext}.
This function basically skips over any marked nodes and returns the next unmarked node.
This effectively ignores logically deleted nodes.
Therefore, running $q$ on a reachable concurrent state $C$ has the same effect as running $q_I$ on $F_I(C)$.
Plugging all this into the proof framework we specified earlier shows that $q$ can be added to Harris's linked list as a \solo{} query.

\begin{figure*}[t]
\begin{lstlisting}[linewidth=.99\columnwidth, numbers=left, frame=none]
struct Node { Key key; Node* next; }

Node* getNext(Node* node) {
  Node* n = node->next;
  while(n != NULL && is_marked(n->next))
    n = unmark(n->next);
  return n; }
\end{lstlisting}
\caption{Implementation of the \var{getNext} operation for Harris Linked List.}
\label{fig:getNext}
\end{figure*}

\begin{figure*}
	\begin{minipage}[t]{.46\textwidth}
		\begin{lstlisting}[linewidth=.99\columnwidth, numbers=left, frame=none]
using @\hl{Value = Node*}@; @\codelineskip@
class Camera {
	int timestamp;
	Camera() { timestamp = 0; }   @\label{line:ss-con-opt}@@\codelineskip@
	int takeSnapshot() {
	// same as in VerCAS}@\codelineskip@
@\hl{Value invalidNextv = new Node();}@
int TBD = -1; @\codelineskip@
class Node {
	/* other fields of
	Node struct of D */
	...
	@\hl{Value nextv;}@ // initially invalidNextv
	@\hl{int ts;}@      // initially TBD }; @\codelineskip@
class OptVersionedCAS {
	@\hl{Value Head;}@
	Camera* S; @\codelineskip@
	OptVersionedCAS(Value v, Camera* s) {
		S = s;
		@\hl{Head = v;}@
		@\hl{if(Head != NULL)}@ {
			@\hl{initNextv(Head);}@
			initTS(Head);  } }@\codelineskip@
	void @\hl{initNextv}@(Value n) {
		if(n->nextv == invalidNextv)
			CAS(&n->nextv, invalidNextv, NULL);}@\codelineskip@
	void initTS(@\hl{Value n}@) {
		if(n->ts == TBD) {
			int curTS = S->timestamp;
			CAS(&n->ts, TBD, curTS); } }
		\end{lstlisting}
	\end{minipage}\hspace{.3in}
	\begin{minipage}[t]{.46\textwidth}
		\StartLineAt{31}
		\begin{lstlisting}[linewidth=.99\textwidth, numbers=left,frame=none]
	Value OptreadSnapshot(int ts) {
		@\hl{Value node = Head;}@
		@\hl{if(node != NULL)}@ initTS(node);
		while(@\hl{node != NULL}@ && node->ts > ts)
			node = node->nextv;
		return @\hl{node}@; }@\codelineskip@
	Value OptvRead() {
		@\hl{Value head = Head;}@
		@\hl{if(node != NULL)}@ initTS(head);
		return @\hl{head}@; }@\codelineskip@
	bool OptvCAS(Value oldV, Value newV) {
		@\hl{Value head = Head;}@
		@\hl{if(node != NULL)}@ initTS(head);
		if (head != oldV) return false;
		if (newV == oldV) return true;
		@\hl{CAS(\&n->nextv, invalidNextv, NULL);}@
		if(CAS(&Head, head, newV)) {
			initTS(newV);
			return true; }
		else {
			// Head can't be NULL
			initTS(Head);
			return false; } } };
		\end{lstlisting}
	\end{minipage}
	\caption{Linearizable implementation of a \vCAS{} object without indirection. Major differences between this and Figure \ref{fig:vcas-alg} are highlighted.
	}
	\label{fig:vcas-direct-alg}
\end{figure*} 

\section{Direct \vCAS\ Algorithm}\label{sec:opt-appendix}

Here we consider the correctness and give pseudocode for the
optimization that avoid indirection.
Consider two \vCAS\ objects $O_1$ and $O_2$.
Roughly speaking, 
since all nodes are recorded-once, every \vcas\ operation on any \vCAS\ 
object stores a distinct value. This means that every time a \vcas\
is executed (on any \vCAS\ object), it writes a pointer to a newly allocated node. Thus, the only
way for a node $nd$ other than the last in the version list of  $O_1$,
to appear in the version list of $O_2$,
is if a pointer to $nd$ were used as the initial value of $O_2$. Hence $nd$ is the last
node in $O_2$'s version list.
We argue that no invocation of \ovrsnap\ on a \vCAS\ object $O$
traverses the \var{nextv} pointer of the last node in the version list of $O$.
These imply that the version lists
of \vCAS\ objects behave as if they are disjoint.
In particular, we never have to store \var{nextv} pointers for two different version lists in the same node.
Based on the above arguments, we can prove that $D_{opt}$ is linearizable. We present pseudo-code for optimized \vCAS{} objects in Figure \ref{fig:vcas-direct-alg}.

\end{document}